\documentclass[11pt,a4paper]{article}
\usepackage{hyperref}
\usepackage{amsmath,amsfonts,amssymb,amsthm,mathtools}
\usepackage{authblk}
\usepackage{latexsym,graphicx}
\usepackage{dsfont}
\usepackage{comment}
\usepackage{amscd}
\usepackage[dvipsnames]{xcolor}
\usepackage{extarrows} 
	\numberwithin{equation}{section} 
\usepackage{yfonts}
\usepackage{authblk}

\usepackage{appendix}

\newcommand{\bphi}{\boldsymbol{\phi}}

\usepackage{todonotes}

\usepackage{tikz}
	 \usetikzlibrary{calc} 
	\usetikzlibrary{matrix,arrows,decorations.pathmorphing} 
	\usepackage{tikz-3dplot}

\allowdisplaybreaks

\def\Xint#1{\mathchoice
{\XXint\displaystyle\textstyle{#1}}%
{\XXint\textstyle\scriptstyle{#1}}%
{\XXint\scriptstyle\scriptscriptstyle{#1}}%
{\XXint\scriptscriptstyle\scriptscriptstyle{#1}}%
\!\int}
\def\XXint#1#2#3{{\setbox0=\hbox{$#1{#2#3}{\int}$}
\vcenter{\hbox{$#2#3$}}\kern-.5\wd0}}

\def\pvint{\Xint-}

\newcommand{\ii}{{\rm i}}

\newcommand{\R}{{\mathbb R}}
\newcommand{\C}{{\mathbb C}}
\newcommand{\Z}{{\mathbb Z}}

\newcommand{\cN}{{\mathcal  N}}

\newtheorem{theorem}{Theorem}
\newtheorem{corollary}{Corollary}[theorem]
\newtheorem{proposition}{Proposition}[section]
\newtheorem{lemma}[proposition]{Lemma}
\newtheorem{remark}{Remark}[section]
\newtheorem{IVP}{Initial value problem}

\newcommand{\ba}{\mathbf{a}}
\newcommand{\bb}{\mathbf{b}}

\newcommand{\bn}{\mathbf{n}}

\newcommand{\bs}{\mathbf{s}}
\newcommand{\bt}{\mathbf{t}}
\newcommand{\bu}{\mathbf{u}}

\newcommand{\bcu}{\mathbf{U}}

\newcommand{\bv}{\mathbf{v}}

\newcommand{\bx}{\mathbf{x}}
\newcommand{\br}{\mathbf{r}}
\newcommand{\by}{\mathbf{y}}

\newcommand{\bz}{\mathbf{z}}

\newcommand{\bS}{\mathbf{S}}
\newcommand{\bT}{\mathbf{T}}

\newcommand{\cT}{\mathcal{T}}

\newcommand{\ta}{\tilde{a}}

\newcommand{\mR}{\mathsf{R}}
\newcommand{\mS}{\mathsf{S}}

\newcommand{\cn}{\mathrm{cn}}
\newcommand{\sn}{\mathrm{sn}}

\newcommand{\trans}{\top}

\newcommand{\Tt}{\tilde{T}}

\newcommand{\im}{\mathrm{Im} }

\newcommand{\dotcirc}{\mathrel{ \!
\tikz[baseline={([yshift=-2.7pt]current bounding box.center)},scale=0.10]{
\draw (0,0) circle (1);
\filldraw[black] (0,0) circle (0.11);
}				
							\! }}
							 
\newcommand{\wedgecirc}{\mathrel{\!
\tikz[baseline={([yshift=-2.7pt]current bounding box.center)},scale=0.10]{
\def\a{40};
\pgfmathsetmacro{\c}{cos(\a) };
\pgfmathsetmacro{\s}{sin(\a) };
\draw (0,0) circle (1);
\draw[line width=0.4pt] (-\c,-\s ) -- (0,1);
\draw[line width=0.4pt] (\c,-\s ) -- (0,1);
}				
							\! }}

\newcommand{\dbb}{ \!\!\dot{\,\,\bb}}
\newcommand{\dbphi}{ \mspace{4mu}\dot{\mspace{-4mu}\bphi} }

\title{\Large{Periodic solutions of the non-chiral intermediate Heisenberg ferromagnet equation described by elliptic spin Calogero-Moser dynamics}}
\date{\vspace{-0.5cm}\small\today\vspace{-0.5cm}}

\author[1]{Bjorn K. Berntson}
\author[2]{Rob Klabbers}

\affil[1]{Department of Physics, KTH Royal Institute of Technology, SE-100 44 Stockholm, Sweden}
\affil[2]{Institut f\"{u}r Physik, Humboldt-Universit\"{a}t zu Berlin,
Zum Gro{\ss}en Windkanal 2, 12489 Berlin, Germany}

\vspace{2mm}

\date{\today}

\begin{document}

\begin{flushright}

	\footnotesize
	HU-EP-22/15 \\
	HU-Mathematik-2022-02
\end{flushright}

\vspace{-15pt}
\bigskip

{\let\newpage\relax\maketitle}

\maketitle

\newcommand{\E}[1]{{\color{blue}{#1}}}

\begin{abstract}
We present a class of periodic solutions of the non-chiral intermediate Heisenberg ferromagnet (ncIHF) equation, which was recently introduced by the authors together with Langmann as a classical, continuum limit of an Inozemtsev-type spin chain. These exact analytic solutions are constructed via a spin-pole ansatz written in terms of certain elliptic functions. The dynamical parameters in our solutions solve an elliptic spin Calogero-Moser (CM) system subject to certain constraints. In the course of our construction, we establish a novel Bäcklund transformation for this constrained elliptic spin CM system.
\end{abstract} 

 

\tableofcontents

\section{Introduction} 

A classic result in the theory of integrable systems \cite{airault1977,choodnovsky1977} states that the soliton dynamics of the Korteweg-de Vries equation is governed by an ($A$-type) Calogero-Moser (CM) system. This relation between two of the best-known integrable systems is but one instance of a \textit{soliton-CM correspondence}, whereby integrable PDEs are linked to many-body systems of CM type. For many such PDEs, including the Korteweg-de Vries \cite{airault1977,choodnovsky1977}, nonlinear Schr\"{o}dinger \cite{hone1997}, Benjamin-Ono \cite{chen1979}, and intermediate long wave \cite[Chapter~3]{ablowitz1981} equations, this is accomplished by making an ansatz for the solution with time-dependent poles in the complex plane and showing that the locations of these poles evolve according to a (complexified) CM system. As such CM systems are exactly-solvable \cite{olshanetsky1983}, this process provides classes of exact analytic solutions to the PDEs. 

A complementary approach is to construct integrable systems with infinite degrees of freedom by taking continuum limits of CM systems. The long-range character of the interactions in the CM system corresponds to nonlocal terms in the continuum description, resulting in partial integro-differential equations of Benjamin-Ono type \cite[Chapter~4]{ablowitz1991}. This concept was pioneered by Abanov, Bettelheim, and Wiegmann \cite{abanov2009}, who showed that the continuum dynamics of the rational CM system is described by Euler hydrodynamic equations that are equivalent to an integro-differential variant of the nonlinear Schr\"{o}dinger equation \cite{matsuno2002}. Recent studies have applied this idea to CM-type systems with spin degrees of freedom, first introduced by Gibbons and Hermsen \cite{gibbons1984}; see also \cite{wojciechowski1985}. The half-wave maps (HWM) equation was derived in \cite{zhou2015} and \cite{lenzmann2018,lenzmann2020} as a continuum limit of a classical Haldane-Shastry spin chain \cite{haldane1988,shastry1988}, a limiting case of the trigonometric spin CM system \cite{polychronakos1993}. Lax integrability and an infinite number of conservation laws were established for the HWM equation in \cite{gerard2018} and multi-soliton solutions were constructed in \cite{berntson2020,matsuno2022}. Moreover, the HWM equation admits a family of periodic solutions governed by a trigonometric spin CM system \cite{berntson2020}. Thus, the HWM equation is linked in two distinct ways to the trigonometric spin CM system. In the present paper, we show that this twofold relation can be lifted to the elliptic setting. 

The non-chiral intermediate Heisenberg ferromagnet (ncIHF) equation is a generalization of the HWM equation related to the elliptic spin CM system. Together with Langmann, we introduced the ncIHF equation in \cite{berntsonklabbers2021} as a continuum limit of a classical Inozemtsev spin chain \cite{inozemtsev1990}; the latter is simultaneously an elliptic generalization of the Haldane-Shastry spin chain and a limiting case of the elliptic spin CM system. It is important to note that the ncIHF equation comes in two related variants: (i) an equation with periodic boundary conditions and (ii) an equation posed on the real line, which may be obtained as an infinite-period limit of the first. In this paper, we study the former, which we call the \textit{periodic ncIHF equation}. Basic integrability results for the ncIHF equation on the real line, where the analysis is technically simpler, have already been obtained in \cite{berntsonklabbers2021}: the (aperiodic) ncIHF equation admits a Lax pair, an infinite number of conservation laws, and multi-soliton solutions governed by the hyperbolic spin CM system. One major result of this paper is that the periodic ncIHF equation admits a family of solutions, analogous to the multi-solitons of the (aperiodic) ncIHF equation, governed by the elliptic spin CM system. As the elliptic spin CM system is exactly-solvable \cite{krichever1995spin}, this gives a new class of exact analytic solutions to the periodic ncIHF equation. 

The periodic ncIHF equation describes the time evolution of two coupled spin densities propagating on the circle of circumference $2\ell>0$; these spin densities are represented by functions\footnote{In this paper, we consider generally complex solutions of the periodic ncIHF equation. See Remark~\ref{rem:complex} for a discussion of this strategy.}  $\bu,\bv: \mathbb{R}\times \mathbb{R}\to \C^3$ of $(x,t)\in \R\times \R$ satisfying $\bu(x+2\ell,t)=\bu(x,t)$, $\bv(x+2\ell,t)=\bv(x,t)$, and $\bu(x,t)^2=\bv(x,t)^2=\rho^2$ for some constant $\rho\in\C$. The periodic ncIHF equation reads 
\begin{equation}
\label{eq:ncIHF}
\begin{split}
\bu_t=&+\bu\wedge T\bu_x - \bu\wedge \tilde{T}\bv_x, \\
\bv_t=& - \bv\wedge T\bv_x+\bv\wedge \tilde{T}\bu_x,
\end{split}
\end{equation}
where $T$ and $\tilde{T}$ are integral operators which act componentwise on three-vectors and are defined by 
\begin{equation} 
\label{eq:TTe}
\begin{split} 
(Tf)(x)\coloneqq &\;\frac{1}{\pi}\pvint_{-\ell}^{\ell}\zeta_1(x'-x;\ell,\ii\delta)f(x')\,\mathrm{d}x', \\
(\Tt f)(x)\coloneqq  &\;\frac{1}{\pi}\int_{-\ell}^{\ell} \zeta_1(x'-x+\ii\delta;\ell,\ii\delta) f(x')\, \mathrm{d}x', 
\end{split} 
\end{equation} 
where the dashed integral indicates a principal value prescription and 
\begin{equation}\label{eq:zeta1}
\zeta_1(z;\ell,\ii\delta)\coloneqq \zeta(z;\ell,\ii\delta)-\frac{\zeta(\ell; \ell,\ii\delta)}{\ell}z \quad (z\in\C),
\end{equation}
with $\zeta(z;\ell,\ii\delta)$ the Weierstrass $\zeta$-function with half-periods $\ell$ and $\ii\delta$ ($\delta>0$). The function $\zeta_1(z;\ell,\ii\delta)$ is $2\ell$-periodic and satisfies
\begin{equation}\label{eq:zetalimits}
\lim_{\ell\to\infty} \zeta_1(z;\ell;\ii\delta)=\frac{\pi}{2\delta}\coth\bigg(\frac{\pi}{2\delta}z\bigg),\qquad \lim_{\ell\to\infty} \zeta_1(z+\ii\delta;\ell;\ii\delta)=\frac{\pi}{2\delta}\tanh\bigg(\frac{\pi}{2\delta}z\bigg).
\end{equation}
The ncIHF equation on the real line \cite{berntsonklabbers2021} is obtained in the $\ell\to\infty$ limit; it is given by \eqref{eq:ncIHF} with the $\ell \to \infty$ limit of the operators \eqref{eq:TTe} obtained using \eqref{eq:zetalimits}. We refer to \cite{berntsonlangmann2021} for further details on the relationship between the periodic and real-line versions of \eqref{eq:TTe}. Similarly, in the limit $\delta\to\infty$, it can be shown that \eqref{eq:ncIHF} reduces to two decoupled HWM equations related by a parity transformation $(\bu,\bv)\to (-P\bv,-P\bu)$, where $(Pf)(x,t)\coloneqq f(-x,t)$. More generally, the non-chirality of the ncIHF equation refers to the invariance of \eqref{eq:ncIHF} under this same parity transformation \cite{berntsonklabbers2021}.

The periodic ncIHF equation generalizes known integrable systems (the HWM and real-line ncIHF equations) and originates from another (the elliptic spin CM system) but a Lax pair for it has not yet been established. While we regard the construction of such a Lax pair as an interesting question for future work, the class of exact solutions presented in this paper provides evidence for the integrability of the periodic ncIHF equation. The construction of these solutions is more involved than that of analogous solutions for the ncIHF equation on the real line \cite{berntsonklabbers2021} due to the presence of elliptic functions in both \eqref{eq:TTe} and the spin-pole ansatz \eqref{eq:ansatz} given below. More specifically, a dynamical background vector is necessitated in the spin-pole ansatz and the resulting spin-pole dynamics must satisfy extra constraints versus the real-line case. To overcome these complications and link the spin-pole dynamics to an elliptic spin CM system, we prove a new B\"{a}cklund transformation for the latter. This B\"{a}cklund transformation is a key result of this paper; we believe it is also of independent interest for its striking difference from known B\"{a}cklund transformations in degenerate cases \cite{gibbons1983,berntsonlangmannlenells2021,berntsonklabbers2020,berntsonklabbers2021}; namely, a new degree of freedom, corresponding to the background vector in the spin-pole ansatz, is required to mediate the transformation between two solutions of spin CM systems.

In the remainder of this section, we focus on stating and describing our two main results, a B\"{a}cklund transformation for the elliptic spin CM system and a class of exact solutions of the periodic ncIHF equation, and describe the organization of the paper. Before proceeding, we introduce notation used in this section and throughout the paper.

\subsection{Notation}
We use the shorthand notation $\sum_{k\neq j}^N$ for sums $\sum_{k=1,k\neq j}^N$, etc. The components of a three-vector $\bs\in\C^3$ are denoted by $(s^1,s^2,s^3)$ and the dot and cross products of two vectors $\bs, \bt\in \C^3$ are defined as
$\bs\cdot\bt=\sum_{a=1}^3 s^at^a$ and $\bs\wedge\bt=(s^2t^3-s^3t^2,s^3t^1-s^1t^3,s^1t^2-s^2t^1)$, respectively. The set of real vectors $\bs \in \R^3$ satisfying $\bs\cdot\bs=1$, i.e., the two-sphere, is denoted by $S^2$. We write the zero vector as $\mathbf{0}=(0,0,0)$.

Dots above a variable indicate differentiation with respect to time while primes indicate differentiation with respect to the argument of a function. Complex conjugation and matrix transposition are denoted by $*$ and $\trans$, respectively.

\subsection{B\"{a}cklund transformation for an elliptic spin CM system}

(Complexified) spin CM systems describe the time evolution of a system of $N\in\Z_{\geq 1}$ particles with internal degrees of freedom moving in the complex plane. We consider the case where the internal degrees of freedom can be represented by complex three vectors, which is a special case of more general systems introduced by Gibbons and Hermsen \cite{gibbons1984} and Wojciechowski \cite{wojciechowski1985}; see \cite{berntsonklabbers2020} for the precise relation. Each particle is represented by a position $a_j=a_j(t)\in\C$ and a spin vector $\bs_j=\bs_j(t)\in \C^3$. We define the elliptic spin CM system to be the following system of equations,
\begin{subequations}\label{eq:sCM1}
\begin{align}
\ddot{a}_j=&\; -2\sum_{k\neq j}^N \bs_j\cdot\bs_k \wp_2'(a_j-a_k) \quad (j=1,\ldots,N), \label{eq:sCMa} \\
\dot{\bs}_j=&\;  -2\sum_{k\neq j}^N \bs_j\wedge\bs_k \wp_2(a_j-a_k) \quad (j=1,\ldots,N), \label{eq:sCMs}
\end{align}
\end{subequations}
where $\wp_2(z)$ is, up to an additive constant, the Weierstrass $\wp$-function with half-periods $\ell$ and $\ii\delta$,
\begin{equation}\label{eq:wp2}
\wp_2(z;\ell,\ii\delta)\coloneqq \wp(z;\ell,\ii\delta)+\frac{\zeta(\ii\delta;\ell,\ii\delta)}{\ii\delta} \quad (z\in\C).
\end{equation}

\begin{remark}\label{rem:equivalent}
Our definition of the elliptic spin CM system \eqref{eq:sCM1} differs from others in the literature, e.g. \cite{krichever1995spin,berntsonlangmannlenells2021}. More specifically, we use the potential $\wp_2(z)$ in place of either $\wp(z)$ or $\wp_1(z)\coloneqq \wp(z)+\zeta(\ell)/\ell$, which differ from $\wp_2(z)$ by additive constants. However, by multiplying each $\bs_j$ in \eqref{eq:sCM1} by an appropriate time-dependent complex rotation $\mR=\mR(t)\in \mathrm{SO}(3;\C)$, the potential $\wp_2(z)$ can be shifted to $\wp_2(z)+c$ for any constant $c\in\C$. A proof of this claim can be found in Appendix~\ref{app:rotation}.
\end{remark}

Elliptic spin CM systems are known to be exactly-solvable \cite{krichever1995spin}, which gives, in principle, exact analytic solutions of the periodic ncIHF equation via our main result, Theorem~\ref{thm:main}, presented below. From our current perspective, the most important property of \eqref{eq:sCM1} is the existence of a Bäcklund transformation relating certain distinct solutions of the elliptic spin CM system; later we will employ this B\"{a}cklund transformation to link the periodic ncIHF equation to the elliptic spin CM system. A Bäcklund transformation valid for the rational, trigonometric, and hyperbolic Gibbons-Hermsen spin CM systems \cite{gibbons1984} was presented in \cite{gibbons1983}; see \cite{berntsonlangmannlenells2021} for a detailed proof. We now describe a Bäcklund transformation for the elliptic spin CM system \eqref{eq:sCM1} subjected to certain constraints which arise in our analysis of the periodic ncIHF equation \eqref{eq:ncIHF}. 

Consider a second elliptic spin CM system for $M\in\Z_{\geq 1}$ particles described by positions $b_j=b_j(t)$ and spin vectors $\bt_j=\bt_j(t)\in\C^3$; the equations of motion read
\begin{subequations}\label{eq:sCM2}
\begin{align}
\ddot{b}_j=&\; -2\sum_{k\neq j}^M \bt_j\cdot\bt_k \wp_2'(b_j-b_k) \quad (j=1,\ldots,M), \label{eq:sCMb} \\
\dot{\bt}_j=&\;  -2\sum_{k\neq j}^M \bt_j\wedge\bt_k \wp_2(b_j-b_k) \quad (j=1,\ldots,M). \label{eq:sCMt}
\end{align}
\end{subequations}
Under appropriate circumstances, solutions of \eqref{eq:sCM1} and \eqref{eq:sCM2} may be related via a system of first-order differential equations involving also a vector $\bphi=\bphi(t)\in\C^3$,
\begin{equation}
\begin{split}\label{eq:ajdot} 
\bs_j\dot{a}_j=&\; -\bs_j\wedge \Bigg(\ii\bphi-\sum_{k\neq j}^N\bs_k\zeta_2(a_j-a_k)+\sum_{k=1}^M \bt_k \zeta_2(a_j-b_k+\ii\delta)\Bigg)  \quad (j=1,\ldots,N),     \\
\bt_j\dot{b}_j=&\; +\bt_j\wedge \Bigg(\ii\bphi+\sum_{k\neq j}^M \bt_k\zeta_2(b_j-b_k)-\sum_{k=1}^N \bs_k \zeta_2(b_j-a_k+\ii\delta)\Bigg) \quad (j=1,\ldots,M)
\end{split}
\end{equation}
and 
\begin{equation}\label{eq:phidot}
\dbphi=\frac{\ii}{2}\sum_{j=1}^N\sum_{k\neq j}^N \bs_j\wedge\bs_k f_2'(a_j-a_k)-\frac{\ii}{2}\sum_{j=1}^M\sum_{k\neq j}^M \bt_j\wedge\bt_k f_2'(b_j-b_k), 
\end{equation}
where
\begin{equation}\label{eq:zeta2}
\zeta_2(z;\ell,\ii\delta)\coloneqq \zeta(z;\ell,\ii\delta)-\frac{\zeta(\ii\delta;\ell,\ii\delta)}{\ii\delta}z \quad (z\in \C)
\end{equation}
(note that $\wp_2(z)=-\zeta_2'(z)$) and
\begin{equation}\label{eq:f2}
f_2(z;\ell,\ii\delta)\coloneqq \zeta_2(z;\ell,\ii\delta)^2-\wp_2(z;\ell,\ii\delta) \quad (z\in\C).
\end{equation}
The precise statement is as follows.

\begin{theorem}\label{thm:backlund}
Let $N,M\in \Z_{\geq 1}$ and $T>0$. Suppose that $\bphi$, $\{a_j,\bs_j\}_{j=1}^N$, and $\{b_j,\bt_j\}_{j=1}^M$ is a solution of the first-order system \eqref{eq:sCMs}, \eqref{eq:sCMt}, \eqref{eq:ajdot}--\eqref{eq:phidot} on the interval $[0,T)$ and that the following constraints hold at $t=0$, 
\begin{align}
&\bs_j^2=0 \quad (j=1,\ldots,N),\qquad  \bt_j^2=0 \quad     (j=1,\ldots,M),   \label{eq:constraint1}\\
\begin{split}
&\bs_j\cdot \Bigg(\ii \bphi- \sum_{k\neq j}^N \bs_k\zeta_2(a_j-a_k) + \sum_{k=1}^M \bt_k \zeta_2(a_j-b_j+\ii\delta) \Bigg)=0 \quad (j=1,\ldots,N), \label{eq:constraint2} \\
&\bt_j\cdot \Bigg(\ii \bphi+ \sum_{k\neq j}^M \bt_k\zeta_2(b_j-b_k) - \sum_{k=1}^N \bs_k \zeta_2(b_j-a_k+\ii\delta) \Bigg)=0 \quad (j=1,\ldots,M),
\end{split} \\
&\sum_{j=1}^N \bs_j-\sum_{j=1}^M \bt_j=\boldsymbol{0}. \label{eq:constraint3}
\end{align}
Then, the second-order equations \eqref{eq:sCMa} and \eqref{eq:sCMb} hold on $[0,T)$.
\end{theorem}

A proof of Theorem~\ref{thm:backlund} is given in Section~\ref{sec:Backlund}. Each of the constraints \eqref{eq:constraint1}-\eqref{eq:constraint3} corresponds to a conserved quantity; if the constraints are satisfied at $t=0$, they also hold at future times when the first-order equations \eqref{eq:phidot}, \eqref{eq:ajdot}, \eqref{eq:sCMs}, and \eqref{eq:sCMt} are satisfied; this fact is proven for \eqref{eq:constraint1} and \eqref{eq:constraint2} in Proposition~\ref{prop:conserved} and for \eqref{eq:constraint3} in Lemma~\ref{lem:totalspin}.

\begin{remark}
The terms $\ii\delta$ appearing in arguments of functions in \eqref{eq:ajdot} and \eqref{eq:constraint2} can be removed by the transformation
\begin{equation}\label{eq:transformation}
a_j\to a_j-\ii\delta/2 \quad (j=1,\ldots,N),\qquad b_j\to b_j+\ii\delta/2 \quad (j=1,\ldots,M),
\end{equation}
using the fact that $\zeta_2(z)$ is $2\ii\delta$-periodic \eqref{eq:imperiod}. The transformation \eqref{eq:transformation} leaves \eqref{eq:sCM1}, \eqref{eq:sCM2}, \eqref{eq:phidot}, \eqref{eq:constraint1}, and \eqref{eq:constraint3} unchanged. We will use Theorem~\ref{thm:backlund} in the proof of Theorem~\ref{thm:main} below; for this application, it is convenient to have the terms $\ii\delta$ in place. 
\end{remark}

\subsection{A class of elliptic solutions of the periodic ncIHF equation}

We construct solutions of the periodic ncIHF equation with dynamics governed by a pair of elliptic spin CM systems, which are related to each other through the B\"{a}cklund transformation of Theorem~\ref{thm:backlund}. More specifically, we make the ansatz for solutions of the periodic ncIHF equation,
\begin{align}\label{eq:ansatz}
\left(\begin{array}{c} \bu(x,t) \\ \bv(x,t)   \end{array}\right)=  \bphi(t)\left(\begin{array}{c} 1 \\ 1    \end{array}\right) &\; +\ii \sum_{j=1}^N \bs_j(t) \left(\begin{array}{c} \zeta_2(x-a_j(t)+\ii\delta/2) \\ \zeta_2(x-a_j(t)-\ii\delta/2)    \end{array}\right) \nonumber\\
&\; -\ii \sum_{j=1}^M \bt_j(t) \left(\begin{array}{c} \zeta_2(x-b_j(t)-\ii\delta/2) \\ \zeta_2(x-b_j(t)+\ii\delta/2)    \end{array}\right),
\end{align}
where $\bphi(t), \bs_j(t),\bt_j(t)\in \C^3$ and $a_j(t),b_j(t)\in \C$ and show that these parameters must satisfy the assumptions of Theorem~\ref{thm:backlund}. In this case, the ansatz \eqref{eq:ansatz} will satisfy $\bu(x,t)^2=\bv(x,t)^2=\rho^2$, for some constant $\rho\in \C$, provided certain constraints on the initial values of the parameters are fulfilled. Theorem~\ref{thm:backlund} and standard results concerning the existence and uniqueness of solutions to systems of ODEs allow us to formulate our result as a relation between (i) certain solutions of the elliptic spin CM systems \eqref{eq:sCM1}, \eqref{eq:sCM2} and background dynamics \eqref{eq:phidot} and (ii) a class of solutions of the periodic ncIHF equation satisfying $\bu(x,t)^2=\bv(x,t)^2=\rho^2$. The precise statement is now given.

\begin{theorem}\label{thm:main}
For $N,M\in \Z_{\geq 1}$ and $T>0$, let $\bphi$, $\{a_j,\bs_j\}_{j=1}^N$, and $\{b_j,\bt_j\}_{j=1}^M$ be a solution of the system of equations \eqref{eq:sCM1}, \eqref{eq:sCM2}, and \eqref{eq:phidot} on the interval $[0,T)$ with initial conditions that satisfy \eqref{eq:ajdot}, \eqref{eq:constraint1}-\eqref{eq:constraint3}, and
\begin{align}\label{eq:constraint4}
\bphi^2=&\; \rho^2+\frac12 \sum_{j=1}^N\sum_{k\neq j}^N \bs_{j}\cdot \bs_{k} f_2(a_{j}-a_{k})+\frac12\sum_{j=1}^M\sum_{k\neq j}^M \bt_{j}\cdot \bt_{k} f_2(b_{j}-b_{k}) \nonumber\\
&\;  -\sum_{j=1}^N\sum_{k=1}^M \bs_{j}\cdot \bt_{k} f_2(a_{j}-b_{k}+\ii\delta)
\end{align}
for some constant $\rho\in \C$ at $t=0$. Moreover, suppose that the conditions
\begin{equation}\label{eq:imaj}
\frac{\delta}{2}<\im\,a_j(t)<\frac{3\delta}{2} \quad (j=1,\ldots,N),\qquad -\frac{3\delta}{2}<\im\,b_j(t)< -\frac{\delta}{2} \quad (j=1,\ldots,M),
\end{equation}
\begin{equation}\label{eq:ajakbjbk}
a_j(t)\neq a_k(t) \quad (1\leq j<k\leq N),\qquad b_j(t)\neq b_k(t) \quad (1\leq j<k\leq M),
\end{equation}
and
\begin{equation}\label{eq:sjtj}
\bs_j\neq \boldsymbol{0} \quad (j=1,\ldots,N),\qquad \bt_j\neq \boldsymbol{0} \quad (j=1,\ldots,M)
\end{equation}
hold for $t\in [0,T)$. Then, for all $t\in [0,T)$ such that the functions $\bu(x,t)$ and $\bv(x,t)$ in \eqref{eq:ansatz} are differentiable with respect to $x$ and $t$ for all $x\in [-\ell,\ell)$, \eqref{eq:ansatz} provides an exact solution of the periodic ncIHF equation \eqref{eq:ncIHF} satisfying $\bu(x,t)^2=\bv(x,t)^2=\rho^2$.
\end{theorem}

\begin{remark}
It is not obvious that the ansatz \eqref{eq:ansatz} is $2\ell$-periodic. The function $\zeta_2(z)$ is $2\ell$-quasi-periodic \eqref{eq:realperiod} and hence $\bu(x+2\ell)-\bu(x)$ and $\bv(x+2\ell)-\bv(x)$ are proportional to $\sum_{j=1}^N \bs_j-\sum_{j=1}^M \bt_j$, i.e., the left hand side of \eqref{eq:constraint3}. We later show that \eqref{eq:constraint3} corresponds to a conserved quantity of the elliptic spin CM system: if it is satisfied at $t=0$, as required in Theorem~\ref{thm:main}, then it holds for $t\in [0,T)$, see Lemma~\ref{lem:totalspin}. The constraint \eqref{eq:constraint3} is also required for $\bu(x,t)^2=\bv(x,t)^2=\rho^2$, see Proposition~\ref{prop:constraints}.
\end{remark}

\begin{remark}\label{rem:complex}
We emphasize that the solutions in Theorem~\ref{thm:main} are generically complex-valued, i.e., $\bu(x,t),\bv(x,t)\in \C^3$ and satisfy $\bu(x,t)^2=\bv(x,t)^2=\rho^2$ for some constant $\rho\in \C$. Real-valued solutions of unit length are described by the consistent reduction $M=N$, $\rho=1$, $\bphi^*=\bphi$, $b_j=a_j^*$, and $\bt_j=\bs_j^*$ of the theorem, which is given as Corollary~\ref{cor:main} in Section~\ref{sec:explicit}, where examples of such solutions are presented. We have chosen our approach because (i) the proofs in the generic, complex case are no more difficult than in the real case and (ii) at least one interesting class of solutions, considered in Section~\ref{subsec:tw}, is necessarily complex: in the case $N=M=1$ of Theorem~\ref{thm:main}, which contains one-soliton, traveling wave solutions in the analogous real-line case \cite[Section~6.1]{berntsonklabbers2021}, there is no solution of the constraints \eqref{eq:constraint1} and \eqref{eq:constraint3} satisfying $\bs_1^*=\bt_1$; all solutions obtained under these conditions from Theorem~\ref{thm:main} are complex.
\end{remark}

\subsection{Plan of the paper}

We prove Theorem~\ref{thm:main} by establishing a sequence of intermediate results including Theorem~\ref{thm:backlund}. In Section~\ref{sec:solitons}, we derive constraints on the parameters in \eqref{eq:ansatz} and we show that the parameters satisfy the first-order system of ODEs of Theorem~\ref{thm:backlund}. We show that this system of ODEs preserves the constraints \eqref{eq:constraint1}--\eqref{eq:constraint3} and \eqref{eq:constraint4} in Section~\ref{sec:conserved}.  In Section~\ref{sec:Backlund}, we prove the Bäcklund transformation, Theorem~\ref{thm:backlund}, in the course of proving Theorem~\ref{thm:main}. Examples of solutions of the ncIHF equation from Theorem~\ref{thm:main} are constructed in Section~\ref{sec:explicit}. Appendix~\ref{app:elliptic} contains identities for the special functions used in the paper. Appendix~\ref{app:rotation} contains a formal statement and proof of the claim in Remark~\ref{rem:equivalent}. 

\section{Constraints and first-order dynamics}\label{sec:solitons}

We derive conditions under which the ansatz \eqref{eq:ansatz} satisfies (i) $\bu(x,t)^2=\bv(x,t)^2=\rho^2$ and (ii) solves \eqref{eq:ncIHF}. The first requirement yields a number of nonlinear constraints on the parameters appearing in \eqref{eq:ansatz}, which are obtained in Section~\ref{subsec:constraints}. In Section~\ref{subsec:firstorder}, we show that when the ansatz \eqref{eq:ansatz} is subjected to one of these constraints and inserted into \eqref{eq:ncIHF}, the latter is reduced to a system of first-order ODEs.

To prove results in this section, we employ certain notation developed in \cite{berntsonklabbers2021}. Given $\C$-valued functions $F_j,G_j$, $j=1,2$, we form two-vectors and define the following product,
\begin{equation}
\left(\begin{array}{c} F_1 \\ F_2  \end{array}\right)\circ \left(\begin{array}{c} G_1 \\ G_2   \end{array}\right)\coloneqq  \left(\begin{array}{c} F_1G_1 \\ F_2G_2\end{array}\right).
\end{equation}
Similarly, we can combine pairs of three-vectors $\ba_j,\bb_j\in \C^3$, $j=1,2$, and define analogs of the dot and wedge products,
\begin{equation}\label{eq:dotwedgecirc}
\left(\begin{array}{c} \ba_1 \\ \ba_2 \end{array}\right)\dotcirc\left(\begin{array}{c} \bb_1 \\ \bb_2 \end{array}\right)=\left(\begin{array}{c} \ba_1\cdot\bb_1 \\ \ba_2\cdot\bb_2 \end{array}\right),\qquad \left(\begin{array}{c} \ba_1 \\ \ba_2 \end{array}\right)\wedgecirc\left(\begin{array}{c} \bb_1 \\ \bb_2 \end{array}\right)=\left(\begin{array}{c} \ba_1\wedge\bb_1 \\ \ba_2\wedge\bb_2 \end{array}\right). 
\end{equation}
By defining
\begin{equation}
\bcu(x,t)=\left(\begin{array}{c} \bu(x,t) \\ \bv(x,t) \end{array}\right)
\end{equation}
and using \eqref{eq:dotwedgecirc} we may write the periodic ncIHF equation \eqref{eq:ncIHF} as
\begin{equation}\label{eq:ncIHF2}
\bcu_t=\bcu\wedgecirc \cT\bcu_x,
\end{equation}
where
\begin{equation}\label{eq:cT}
\cT=\left(\begin{array}{cc} T & -\tilde{T} \\ \tilde{T} & -T     \end{array}\right), \qquad \cT:  \left(\begin{array}{c} \bu_x \\ \bv_x\end{array}\right)\mapsto  \left(\begin{array}{c} T\bu_x-\tilde{T}\bv_x \\ \tilde{T}\bu_ x-T\bv_x \end{array}\right)
\end{equation}
with $T$ and $\tilde{T}$ as defined in \eqref{eq:TTe}. 

It is also useful to write the ansatz \eqref{eq:ansatz} using this two-vector notation. We define 
\begin{equation}\label{eq:EA}
E\coloneqq \left(\begin{array}{c} 1 \\ 1\end{array}\right),\qquad A_{\pm}(z)\coloneqq \left(\begin{array}{c} \zeta_2(z\pm\ii\delta/2) \\ \zeta_2(z\mp\ii\delta/2) \end{array}\right) \quad (z\in \C),
\end{equation}
so that \eqref{eq:ansatz} can be written as
\begin{equation}\label{eq:ansatzSH}
\bcu(x,t)=\bphi(t) E+\sum_{j=1}^{\cN} r_j \bs_j(t) A_{r_j}(x-a_j(t)),
\end{equation}
using also the shorthand notation
\begin{equation}\label{eq:shorthand}
(a_j,\bs_j,r_j)\coloneqq \begin{cases}
(a_j,\bs_j,+) & j=1,\ldots,N, \\
(b_j,\bt_j,-) & j=N+1,\ldots,\cN,
\end{cases} \qquad \cN=N+M.
\end{equation}

\subsection{Constraints} \label{subsec:constraints}

The following proposition establishes the conditions required for the functions in ansatz \eqref{eq:ansatz} to have constant length.

\begin{proposition}\label{prop:constraints}
The functions $\bu(x,t)$ and $\bv(x,t)$ in \eqref{eq:ansatz} satisfy $\bu(x,t)^2=\bv(x,t)^2=\rho^2$ if and only if the parameters $\bphi$, $\{a_j,\bs_j\}_{j=1}^N$, and $\{b_j,\bt_j\}_{j=1}^M$ satisfy the conditions \eqref{eq:constraint1}--\eqref{eq:constraint3} and \eqref{eq:constraint4}.
\end{proposition}

\begin{proof}

Using \eqref{eq:ansatzSH} and \eqref{eq:dotwedgecirc}, we compute
\begin{align}\label{eq:UdotU1}
\bcu\dotcirc\bcu= &\; \bphi^2 E+2\ii \bphi\cdot \sum_{j=1}^{\cN} r_j\bs_jA_{r_j}(x-a_j)-\sum_{j=1}^{\cN}\sum_{k=1}^{\cN} r_j r_k \bs_j\cdot\bs_k A_{r_j}(x-a_j)\circ A_{r_k}(x-a_k).
\end{align}
To proceed, we need the identities
\begin{equation}\label{eq:Aj2Id}
A_{r_j}(x-a_j)\circ A_{r_j}(x-a_j)=-A_{r_j}'(x-a_j)+F_{r_j}(x-a_j)
\end{equation}
and
\begin{align}\label{eq:AjAkId}
A_{r_j}(x-a_j)\circ A_{r_k}(x-a_k)=&\; \zeta_2(\ta_j-\ta_k)\big(A_{r_j}(x-a_j)-A_{r_k}(x-a_k)\big) \nonumber \\
&\; +\frac12 \big(F_{r_j}(x-a_j)+F_{r_k}(x-a_k)\big)+\frac12 f_2(\ta_j-\ta_k)E +\frac{3\zeta(\ii \delta)}{2\delta}E,
\end{align}
where
\begin{equation}\label{eq:at}
\tilde{a}_j\coloneqq a_j-\ii r_j\delta/2 \quad (j=1,\ldots,\cN)
\end{equation}
and
\begin{equation}\label{eq:F}
F_{\pm}(z)\coloneqq \left(\begin{array}{c} f_2(z\pm\ii\delta/2) \\ f_2(z\mp\ii\delta/2) \end{array}\right) \quad (z\in\C).
\end{equation}
The identities \eqref{eq:Aj2Id} and \eqref{eq:AjAkId} follow from the elliptic identities \eqref{eq:IdV} and \eqref{eq:Idmain}, respectively together with the definitions of $E$, $A_{\pm}(z)$ \eqref{eq:EA} and $F_{\pm}(z)$ \eqref{eq:F}. We evaluate the double sum in \eqref{eq:UdotU1} using \eqref{eq:Aj2Id} for $j=k$ and \eqref{eq:AjAkId} for $j\neq k$:
\begin{align}\label{eq:UdotU2}
\bcu\dotcirc\bcu=&\; \bphi^2 E +2\ii \bphi\cdot \sum_{j=1}^{\cN} r_j\bs_j A_{r_j}(x-a_j)+\sum_{j=1}^{\cN} \bs_j^2 A_{r_j}'(x-a_j)-\sum_{j=1}^{\cN} \bs_j^2 F_{r_j}(x-a_j) \nonumber\\
&\; -\sum_{j=1}^{\cN}\sum_{k\neq j}^{\cN} r_j r_k \bs_j\cdot\bs_k \zeta_2(\ta_j-\ta_k) \big(A_{r_j}(x-a_j)-A_{r_k}(x-a_k)\big) \nonumber\\
&\; -\frac12 \sum_{j=1}^{\cN}\sum_{k\neq j}^{\cN} r_jr_k\bs_j\cdot\bs_k \big(F_{r_j}(x-a_j)+F_{r_k}(x-a_k)+f_2(\ta_j-\ta_k)E \big) \nonumber \\
&\;
-\frac{3\zeta(\ii \delta)}{2\delta}\sum_{j=1}^{\cN}\sum_{k\neq j}^{\cN} r_j r_k\bs_j\cdot\bs_k E.
\end{align}
Next, we recall that the functions $\zeta_2(z)$ and $f_2(z)$ appearing in $A_{\pm}(z)$ and $F_{\pm}(z)$ are odd and even, respectively \eqref{eq:parity}. Using this symmetry to rewrite the double sums in the second and third lines of \eqref{eq:UdotU2} and collecting terms, we find
\begin{align}\label{eq:UdotU3}
\bcu\dotcirc\bcu=&\; \bphi^2 E+2\ii \bphi\cdot \sum_{j=1}^{\cN} r_j\bs_j A(x-\ta_j)+\sum_{j=1}^{\cN} \bs_j^2 A_{r_j}'(x-a_j)-\sum_{j=1}^{\cN} \bs_j^2 F_{r_j}(x-a_j) \nonumber\\
&\; -2\sum_{j=1}^{\cN}\sum_{k\neq j}^{\cN} r_j r_k \bs_j\cdot\bs_k \zeta_2(\ta_j-\ta_k)A_{r_j}(x-a_j) - \sum_{j=1}^{\cN}\sum_{k\neq j}^{\cN}  r_jr_k\bs_j\cdot\bs_k F_{r_j}(x-a_j) \nonumber \\
&\; -\frac12\sum_{j=1}^{\cN}\sum_{k\neq j}^{\cN}      r_jr_k\bs_j\cdot\bs_k f_2(\ta_j-\ta_k) E -\frac{3\zeta(\ii \delta)}{2\delta}\sum_{j=1}^{\cN}\sum_{k\neq j}^{\cN}    r_j r_k\bs_j\cdot\bs_k E \nonumber\\
=&\;  \Bigg( \bphi^2-\frac12\sum_{j=1}^{\cN}\sum_{k\neq j}^{\cN} r_jr_k\bs_j\cdot\bs_k f_2(\ta_j-\ta_k)-\frac{3\zeta(\ii \delta)}{2\delta}\sum_{j=1}^{\cN} \sum_{k\neq j}^{\cN} r_j r_k\bs_j\cdot\bs_k\Bigg) E \nonumber\\
&\; +2\sum_{j=1}^{\cN}r_j\bs_j \cdot \Bigg(\ii\bphi-\sum_{k\neq j}^{\cN} r_k\bs_k \zeta_2(\ta_j-\ta_k)    \Bigg) A_{r_j}(x-a_j) \nonumber \\
&\; -\sum_{j=1}^{\cN} \bs_j^2 A_{r_j}'(x-a_j)-\sum_{j=1}^{\cN}\Bigg(\bs_j^2-\sum_{k\neq j}^{\cN} r_jr_k\bs_j\cdot \bs_k\Bigg) F_{r_j}(x-a_j).
\end{align}
We set \eqref{eq:UdotU3} equal to $\rho^2E$ and note that $E$, $\{A(x-a_j)\}_{j=1}^{\cN}$, $\{A'(x-a_j)\}_{j=1}^{\cN}$, and $\{F(x-a_j)\}_{j=1}^{\cN}$ are linearly independent as a consequence of \eqref{eq:imaj}--\eqref{eq:ajakbjbk}. The terms proportional to $A'_{r_j}(x-a_j)$ and $A_{r_j}(x-a_j)$ give the conditions
\begin{equation}
\bs_j^2=0 \quad (j=1,\ldots,\cN), \label{eq:constraint1SH}
\end{equation}
and
\begin{equation}
\bs_j\cdot\Bigg(\bphi+\ii\sum_{k\neq j}^{\cN} r_k \bs_k\zeta_2(\ta_j-\ta_k)\Bigg)=0 \quad (j=1,\ldots,\cN), \label{eq:constraint2SH}
\end{equation}
respectively. By inserting \eqref{eq:constraint1SH} into the sum in $F_{r_j}(x-a_j)$ in \eqref{eq:UdotU3}, we obtain
\begin{align}\label{eq:constraint3SH}
\sum_{j=1}^{\cN} r_j\bs_j=\boldsymbol{0},
\end{align}
and by inserting \eqref{eq:constraint3SH} into the sum proportional to $E$ in \eqref{eq:UdotU3}, we obtain
\begin{equation}\label{eq:constraint4SH}
\bphi^2-\frac12 \sum_{j=1}^{\cN} \sum_{k\neq j}^{\cN} r_jr_k \bs_j\cdot\bs_k f_2(\ta_j-\ta_k)=\rho^2.
\end{equation}
The constraints \eqref{eq:constraint1SH}--\eqref{eq:constraint4SH} are seen to be equivalent to (\ref{eq:constraint1}-\ref{eq:constraint3}) and \eqref{eq:constraint4} after recalling the notation \eqref{eq:shorthand} and \eqref{eq:at}.
\end{proof}

\subsection{First-order dynamics}\label{subsec:firstorder}

The following proposition describes conditions, in the form of a system of ODEs, when the ansatz \eqref{eq:ansatz} solves the periodic ncIHF equation \eqref{eq:ncIHF} (without the requirement that $\bu(x,t)^2=\bv(x,t)^2=\rho^2$). 

\begin{proposition}\label{prop:firstorder}
Suppose that $\bphi$, $\{a_j,\bs_j\}_{j=1}^N$, and $\{b_j,\bt_j\}_{j=1}^M$ is a solution of the system of equations \eqref{eq:sCMs}, \eqref{eq:sCMt}, and \eqref{eq:ajdot}--\eqref{eq:phidot} on an interval $[0,T)$
with initial conditions that satisfy \eqref{eq:constraint3} at $t=0$. Moreover, suppose that \eqref{eq:imaj} and \eqref{eq:ajakbjbk} hold for all $t\in [0,T)$. Then, for $t\in[0,T)$ such that the functions $\bu(x,t),\bv(x,t)$ in \eqref{eq:ansatz} are differentiable with respect to $x$ and $t$ for all $x\in [-\ell,\ell)$, \eqref{eq:ansatz} provides an exact solution of the periodic ncIHF equation \eqref{eq:ncIHF}.
\end{proposition}

\begin{proof}

We compute both terms in the periodic ncIHF equation in the form \eqref{eq:ncIHF2}, again making use of the form \eqref{eq:ansatzSH} of the pole ansatz \eqref{eq:ansatz} and the shorthand notation \eqref{eq:shorthand} and \eqref{eq:at}. First, we have
\begin{align}\label{eq:Ut}
\bcu_t=&\; \dbphi E +\ii \sum_{j=1}^{\cN} r_j\big(\dot{\bs}_j  A_{r_j}(x-a_j) -\bs_j\dot{a}_j  A_{r_j}'(x-a_j) \big),
\end{align}
using that $\dot{\tilde{a}}_j=\dot{a}_j$. We compute the remaining term in \eqref{eq:ncIHF2} in several steps. To compute $\cT \bcu_x$, we use the fact that the $A_{r_j}'(x-a_j)$ are eigenfunctions of $\cT$ when \eqref{eq:imaj} holds,
\begin{equation}\label{eq:cTA}
(\cT A_{r_j}'(\cdot-a_j))(x)=-\ii r_j A_{r_j}'(x-a_j) \quad (j=1,\ldots,\cN).
\end{equation}
The identity \eqref{eq:cTA} is established by verifying that the functions $\wp_2(z-a_j)$ appearing in $A_{\pm}'(z)$ \eqref{eq:EA} satisfy the conditions of the following result proved in \cite[Appendix A]{berntson2020}:\footnote{Note that the definition of $\cT$ in this paper differs from that in \cite{berntson2020} by a similarity transformation, $\cT\to D\cT D^{-1}$ with $D\coloneqq \mathrm{diag}(1,-1)$; we have modified the statement of the result in \cite{berntson2020} accordingly to meet our present needs.} \textit{for a $2\ell$-periodic function $g(z)$ analytic in a strip $-d<\im\,z<d$ with $d>\delta/2$ and satisfying $\int_{-\ell}^{\ell}g(x)\,\mathrm{d}x=0$, the functions $G_{\pm}(x)\coloneqq (g(x\pm\ii\delta/2),g(x\mp\ii\delta/2))^{\trans}$} are eigenfunctions of the operator $\cT$ with eigenvalues $\mp \ii$. 

Differentiating \eqref{eq:ansatzSH} with respect to $x$ gives
\begin{equation}\label{eq:Ux}
\bcu_x=\ii\sum_{j=1}^{\cN} r_j \bs_j A_{r_j}'(x-a_j)
\end{equation}
and hence \eqref{eq:cTA} implies
\begin{equation}\label{eq:TUx}
\cT \bcu_x=\ii\sum_{j=1}^{\cN} r_j\bs_j\cT A_{r_j}'(x-a_j)=\sum_{j=1}^{\cN} \bs_j A_{r_j}'(x-a_j).
\end{equation}

We compute $\bcu\wedgecirc\cT\bcu_x$ by combining \eqref{eq:ansatzSH} with \eqref{eq:TUx}:
\begin{align}\label{eq:UwedgeTUx1}
\bcu \wedgecirc \cT\bcu_x=&\;  \bphi E  \wedgecirc \sum_{j=1}^N \bs_j  A_{r_j}'(x-a_j)+\ii\sum_{j=1}^N r_j \bs_j A(x-a_j)  \wedgecirc \sum_{k=1}^N \bs_k  A_{r_k}(x-a_k)  \nonumber\\
=&\; \sum_{j=1}^N \bphi\wedge\bs_j A_{r_j}'(x-a_j)  +\ii\sum_{j=1}^N\sum_{k\neq j}^N r_j\bs_j\wedge\bs_k A_{r_j}(x-a_j)\circ A_{r_k}'(x-a_k).
\end{align}
Differentiating \eqref{eq:AjAkId} with respect to $\ta_k$ gives 
\begin{align}\label{eq:ArjArk}
A_{r_j}(x-a_j)\circ A_{r_k}'(x-a_k)=&\; -\wp_2(\ta_j-\ta_k)\big(A_{r_j}(x-a_j)-A_{r_k}(x-a_k)\big)-\zeta_2(\ta_j-\ta_k)A_{r_k}'(x-a_k) \nonumber \\
&\;+\frac12 F_{r_k}'(x-a_k) +\frac12 f_2(\ta_j-\ta_k)E;
\end{align}
inserting this identity into \eqref{eq:UwedgeTUx1}, we have
\begin{align}\label{eq:UwedgeTUx2}
\bcu\wedgecirc\cT\bcu_x=&\;  \sum_{j=1}^{\cN} \bphi\wedge\bs_j A_{r_j}(x-a_j)  -\ii \sum_{j=1}^{\cN}\sum_{k\neq j}^{\cN} r_j\bs_j\wedge\bs_k \wp_2(\ta_j-\ta_k)\big(A_{r_j}(x-a_j)-A_{r_k}(x-a_k)\big)\nonumber \\
&\; -\ii\sum_{j=1}^{\cN}\sum_{k\neq j}^{\cN}r_j\bs_j\wedge\bs_k \zeta_2(\ta_j-\ta_k)A_{r_k}'(x-a_k)+\frac{\ii}{2}\sum_{j=1}^{\cN}\sum_{k\neq j}^{\cN}r_j\bs_j\wedge\bs_k F_{r_k}'(x-a_k) \nonumber \\
&\; +\frac{\ii}{2}\sum_{j=1}^{\cN}\sum_{k\neq j}^{\cN}r_j\bs_j\wedge\bs_k f_{2}(\ta_j-\ta_k)E. 
\end{align}
Next, since $\wedge$ is antisymmetric and $\wp_2(z)$ is an even function \eqref{eq:parity}, we can rewrite the double sum in the first line of \eqref{eq:UwedgeTUx2} according to
\begin{equation}
\begin{aligned}
&\sum_{j=1}^{\cN}\sum_{k\neq j}^{\cN}r_j\bs_j\wedge\bs_k \wp_2(\ta_j-\ta_k)\big(A_{r_j}(x-a_j)-A_{r_k}(x-a_k)\big) \\
&= \frac12\sum_{j=1}^{\cN}\sum_{k\neq j}^{\cN}(r_j+r_k)\bs_j\wedge\bs_k \wp_2(\ta_j-\ta_k)\big(A_{r_k}(x-a_j)-A_{r_k}(x-a_k)\big)  \\
&= \sum_{j=1}^{\cN}\sum_{k\neq j}^{\cN}(r_j+r_k)\bs_j\wedge\bs_k \wp_2(\ta_j-\ta_k)A_{r_j}(x-a_j).
\end{aligned}
\end{equation}
Hence, inserting this and swapping some indices $j\leftrightarrow k$ (using the antisymmetry of $\wedge$ and the fact that $\zeta_2(z)$ is an odd function \eqref{eq:parity}) in \eqref{eq:UwedgeTUx2}, we obtain
\begin{align}\label{eq:UwedgeTUx3}
\bcu\wedgecirc\cT\bcu_x=&\;  \sum_{j=1}^{\cN} \bphi\wedge\bs_j A_{r_j}'(x-a_j)  -\ii \sum_{j=1}^{\cN}\sum_{k\neq j}^{\cN} (r_j+r_k)\bs_j\wedge\bs_k \wp_2(\ta_j-\ta_k)A_{r_j}(x-a_j) \nonumber \\
&\; -\ii\sum_{j=1}^{\cN}\sum_{k\neq j}^{\cN}r_k\bs_j\wedge\bs_k \zeta_2(\ta_j-\ta_k)A_{r_j}(x-a_j)-\frac{\ii}{2}\sum_{j=1}^{\cN}\sum_{k\neq j}^{\cN}r_k\bs_j\wedge\bs_k F_{r_j}'(x-a_j) \nonumber \\
&\; +\frac{\ii}{2}\sum_{j=1}^{\cN}\sum_{k\neq j}^{\cN}r_j\bs_j\wedge\bs_k f_{2}(\ta_j-\ta_k)E,
\end{align}
which may be rearranged to 
\begin{align}\label{eq:UwedgeTUx4}
\bcu\wedgecirc\cT\bcu_x=&\;  \frac{\ii}{2}\sum_{j=1}^{\cN}\sum_{k\neq j}^{\cN}r_j\bs_j\wedge\bs_k f_{2}(\ta_j-\ta_k)E -\ii \sum_{j=1}^{\cN}\sum_{k\neq j}^{\cN} (r_j+r_k)\bs_j\wedge\bs_k \wp_2(\ta_j-\ta_k)A_{r_j}(x-a_j)     \nonumber \\
&\; - \sum_{j=1}^{\cN} \bs_j \wedge \Bigg(\bphi + \ii \sum_{k=1}^{\cN} r_k \bs_k \zeta_2(\ta_j-\ta_k)\Bigg) A_{r_j}'(x-a_j) -\frac{\ii}{2}\sum_{j=1}^{\cN}r_j \Bigg(\sum_{k=1}^{\cN}r_k\bs_k\Bigg) F_{r_j}'(x-a_j),
\end{align}
where we have used $\bs_k\wedge\bs_k=\mathbf{0}$ to rewrite the final sum. 

Inserting \eqref{eq:Ut} and \eqref{eq:UwedgeTUx4} into \eqref{eq:ncIHF2} and using the linear independence of $E$, $\{A_{r_j}(x-a_j)\}_{j=1}^{\cN}$, $\{A_{r_j}'(x-a_j)\}_{j=1}^{\cN}$, and $\{F_{r_j}'(x-a_j)\}_{j=1}^{\cN}$ as a consequence of \eqref{eq:imaj}--\eqref{eq:ajakbjbk}, we obtain the equations of motion 
\begin{align}
\dbphi=&\; \frac{\ii}{4}\sum_{j=1}^{\cN}\sum_{k\neq j}^{\cN} (r_j+r_k)\bs_j\wedge\bs_k f_2'(\ta_j-\ta_k), \label{eq:phidotSH}\\
\dot{\bs}_j=&\; -\sum_{k\neq j}^{\cN} (1+r_jr_k)\bs_j\wedge\bs_k \wp_2(\ta_j-\ta_k) \quad (j=1,\ldots,\cN), \label{eq:sjdotSH}\\
\dot{a}_j \bs_j=&\; -r_j \bs_j \wedge \Bigg(\ii\bphi-\sum_{k\neq j}^{\cN} r_k \bs_k \zeta_2(\ta_j-\ta_k)\Bigg)    \quad (j=1,\ldots,\cN) \label{eq:ajdotSH}
\end{align}
and the constraint \eqref{eq:constraint3SH}. Equations \eqref{eq:phidotSH}--\eqref{eq:ajdotSH} are equivalent to \eqref{eq:phidot}, \eqref{eq:sCMs} and \eqref{eq:sCMt}, and \eqref{eq:ajdot}, respectively via the notation \eqref{eq:shorthand} and \eqref{eq:at}. To prove the proposition, it remains to show that if \eqref{eq:constraint3SH}, equivalent to \eqref{eq:constraint3}, is satisfied at $t=0$, it holds on $[0,T)$ when the variables $\{a_j,\bs_j\}_{j=1}^{\cN}$ evolve according to \eqref{eq:phidotSH}--\eqref{eq:ajdotSH}. We prove the following stronger result.

\begin{lemma}\label{lem:totalspin}
The total spins
\begin{equation}\label{eq:totalspin}
\bS\coloneqq \sum_{j=1}^N \bs_j, \qquad \bT \coloneqq \sum_{j=1}^M \bt_j
\end{equation}
are conserved by the equations of motion \eqref{eq:sCMs} and \eqref{eq:sCMt}. 
\end{lemma}
\begin{proof}
We differentiate $\bS$ in \eqref{eq:totalspin} with respect to $t$ and insert \eqref{eq:sCMs} to find
\begin{align*}
\bS_t=\sum_{j=1}^{N} \dot{\bs}_j=-2\sum_{j=1}^N\sum_{k\neq j}^N \bs_j\wedge \bs_k \wp_2(a_j-a_k).
\end{align*}
The sum vanishes because $\wp_2(z)$ is an even function \eqref{eq:parity} and hence the summand is antisymmetric under the interchange of $j$ and $k$. The proof for $\bT$ is similar. 
\end{proof}
Because $\bS$ and $\bT$ are conserved quantities, so is their difference and hence, \eqref{eq:constraint3} holds on $[0,T)$ if it is satisfied at $t=0$. 
\end{proof}

\section{Conserved quantities}\label{sec:conserved}

This section is devoted to proving that the constraints in Proposition~\ref{prop:constraints} correspond to conserved quantities of the ODE system of Proposition~\ref{prop:firstorder}, i.e., if the constraints are satisfied at $t=0$, as required by Theorem~\ref{thm:main}, they hold at all future times. We note that the constancy of the total spins appearing in the constraint \eqref{eq:constraint3} was already proved in Lemma~\ref{lem:totalspin}.

We prove the following.

\begin{proposition}\label{prop:conserved}
Under the assumptions of Proposition~\ref{prop:firstorder}, the following quantities are conserved:
\begin{align}
&P_j\coloneqq \bs_j^2 \quad (j=1,\ldots,N),\qquad  P_{N+j}\coloneqq \bt_j^2 \quad (j=1,\ldots,M), \label{eq:Pj}\\
\begin{split}\label{eq:Qj}
&Q_j\coloneqq\bs_j\cdot\Bigg(\ii\bphi-\sum_{k\neq j}^{N}\bs_k\zeta_2(a_j-a_k)+\sum_{k=1}^M \bt_k\zeta_2(a_j-b_k+\ii\delta)\Bigg) \quad (j=1,\ldots,N),\\
&Q_{N+j}\coloneqq \bt_j\cdot\Bigg(\ii\bphi+\sum_{k\neq j}^{M} \bt_k\zeta_2(b_j-b_k)-\sum_{k=1}^N \bs_k\zeta_2(b_j-a_k+\ii\delta)\Bigg) \quad (j=1,\ldots,M), 
\end{split} \\
&R\coloneqq \bphi^2-\frac12\sum_{j=1}^N\sum_{k\neq j}^N \bs_j\cdot\bs_k f_2(a_j-a_k)-\frac12\sum_{j=1}^M\sum_{k\neq j}^M \bt_j\cdot\bt_k f_2(b_j-b_k) \nonumber \\
&\phantom{R\coloneqq} +\sum_{j=1}^N\sum_{k=1}^M \bs_j\cdot\bs_k f_2(a_j-b_k+\ii\delta).    \label{eq:Rj}
\end{align}
\end{proposition}

\subsection{Proof of Proposition~\ref{prop:conserved}}

We prove the proposition in three parts corresponding to the quantities \eqref{eq:Pj}, \eqref{eq:Qj}, and \eqref{eq:Rj}. 

\subsubsection{Conservation of $P_j$}

Using the notation \eqref{eq:shorthand}, we write \eqref{eq:Pj} as 
\begin{equation}\label{eq:Pj2}
P_j=\bs_j^2 \quad (j=1,\ldots,\cN).
\end{equation}
Differentiating this with respect to time and inserting \eqref{eq:sjdotSH}, we have
\begin{align}
\dot{P}_{j}=  \bs_j\cdot\dot{\bs}_j =&\; -\sum_{k=1}^{\cN}(1+r_jr_k) \bs_j\cdot (\bs_j\wedge\bs_k)\wp_2(\ta_j-\ta_k) \nonumber \\
=&\; -\sum_{k=1}^{\cN}(1+r_jr_k) \bs_k\cdot (\bs_j\wedge\bs_j)\wp_2(\ta_j-\ta_k)=0,
\end{align}
where we have used the invariance of the vector triple product under cyclic permutations,
\begin{equation}\label{eq:cyclic}
\bx\cdot(\by\wedge\bz)=\bz\cdot(\bx\wedge\by)=\by\wedge(\bz\wedge\bx).
\end{equation}

\subsubsection{Conservation of $Q_j$}

We write \eqref{eq:Qj} as
\begin{equation}\label{eq:Yj2}
Q_j=\bs_j\cdot\bb_j \quad (j=1,\ldots,\cN),
\end{equation}
where
\begin{equation}\label{eq:bj}
\bb_j\coloneqq \ii\bphi-\sum_{k\neq j}^{\cN} r_k\bs_k \zeta_2(\ta_j-\ta_k) \quad (j=1,\ldots,\cN),
\end{equation}
with the shorthand notation \eqref{eq:shorthand} and \eqref{eq:at}.

Differentiating \eqref{eq:Yj2} with respect to time and inserting \eqref{eq:sjdotSH}--\eqref{eq:ajdotSH}, we find
\begin{align}\label{eq:Qjdot1}
\dot{Q}_{j}=&\; \dot{\bs}_j\cdot\bb_j+\bs_j\cdot\dbb_j \nonumber \\
=&\; -\sum_{j=1}^{\cN} (1+r_jr_k)(\bs_j\wedge\bs_k)\cdot\bb_j \wp_2(\ta_j-\ta_k) \nonumber \\
&\;+\bs_j\cdot\Bigg(\ii\dbphi-\sum_{k\neq j}^{\cN} r_k \dot{\bs}_k\zeta_2(\ta_j-\ta_k)+\sum_{k\neq j}^{\cN} r_k \bs_k (\dot{a}_j-\dot{a}_k)\wp_2(\ta_j-\ta_j)\Bigg) \nonumber \\
=&\; -\sum_{j=1}^{\cN} (1+r_jr_k)(\bs_j\wedge\bs_k)\cdot\bb_j            \wp_2(\ta_j-\ta_k) \nonumber \\
&\;+\ii \bs_j\cdot \dbphi-\sum_{k\neq j}^{\cN}\sum_{l\neq k}^{\cN} r_k(1+r_kr_l) \bs_j\cdot (\bs_k\wedge\bs_l) \zeta_2(\ta_j-\ta_k)\wp_2(\ta_k-\ta_l) \nonumber \\
&\; +\sum_{k\neq j}^{\cN} r_k(-r_j (\bs_j\wedge\bb_j)\cdot{\bs_k}+r_k(\bs_k\wedge\bb_k)\cdot\bs_j) \wp_2(\ta_j-\ta_j) .
\end{align}
We use \eqref{eq:cyclic} again to reorder triple products,
\begin{align}\label{eq:Qjdot2}
\dot{Q}_j =&\; -\sum_{k\neq j}^{\cN} \big((1+r_jr_k) \bs_j\cdot (\bs_j\wedge\bs_k)+r_jr_k (\bs_j\wedge\bb_j)\cdot\bs_k-(\bs_k\wedge\bb_k)\cdot\bs_j\big)\wp_2(\ta_j-\ta_k) \nonumber\\
&\; +\ii \bs_j\cdot \dbphi-\sum_{k\neq j}^{\cN}\sum_{l\neq k}^{\cN} (r_k+r_l) \bs_j\cdot(\bs_k\wedge\bs_l)\zeta_2(\ta_j-\ta_k)\wp_2(\ta_k-\ta_l) \nonumber \\ 
=&\; -\sum_{k\neq j}^{\cN} (\bs_j\wedge\bs_k)\cdot(\bb_j-\bb_k)\wp_2(\ta_j-\ta_k) \nonumber \\
&\; +\ii\bs_j\cdot\dbphi-\sum_{k\neq j}^{\cN}\sum_{l\neq k}^{\cN} (r_k+r_l) \bs_j\cdot(\bs_k\wedge\bs_l)\zeta_2(\ta_j-\ta_k)\wp_2(\ta_k-\ta_l).
\end{align}

To proceed, we rewrite the quantity $(\bb_j-\bb_k)\wp_2(\ta_j-\ta_k)$ in a convenient way. By the definition of $\bb_j$ \eqref{eq:bj},
\begin{align}\label{eq:bjbk1} 
(\bb_j-\bb_k)\wp_2(\ta_j-\ta_k)=&\; -\Bigg(\sum_{l\neq j}^{\cN} r_l \bs_l \zeta_2(\ta_j-\ta_l)- \sum_{l\neq k}^{\cN} r_l \bs_l \zeta_2(\ta_k-\ta_l)\Bigg)\wp_2(\ta_j-\ta_k) \nonumber \\
=&\;  -(r_k\bs_k+r_j\bs_j)\zeta_2(\ta_j-\ta_k)\wp_2(\ta_j-\ta_k)           \nonumber \\
&\; -\sum_{l\neq j,k}^{\cN} r_l \bs_l \big(\zeta_2(\ta_j-\ta_l)-\zeta_2(\ta_k-\ta_l)\big)\wp_2(\ta_j-\ta_k),
\end{align}
where we have used the fact that $\zeta_2(z)$ is an odd function \eqref{eq:parity} in the second step. To proceed, we use the identities
\begin{equation}\label{eq:EllipticId1}
\zeta_2(z)\wp_2(z)=-\frac12\big(\wp_2'(z)+f_2'(z)\big)
\end{equation}
 and
\begin{align}\label{eq:EllipticId2}
\big(\zeta_2(\ta_j-\ta_l)-\zeta_2(\ta_k-\ta_l)\big)\wp_2(\ta_j-\ta_k)=&\; -\big(\zeta_2(\ta_j-\ta_k)-\zeta_2(\ta_j-\ta_l)\big)\wp_2(\ta_k-\ta_l) \nonumber \\
&\; -\frac12\big(f_2'(\ta_j-\ta_k)-f_2'(\ta_k-\ta_l)\big),
\end{align}
The first identity \eqref{eq:EllipticId1} is obtained by differentiating \eqref{eq:IdV} with respect to $z$ and the second identity \eqref{eq:EllipticId2} is obtained by differentiating \eqref{eq:Idmain} with respect to $a$ and setting $z=\ta_j$, $a=\ta_k$, and $b=\ta_l$.

Inserting \eqref{eq:EllipticId1} and \eqref{eq:EllipticId2} into \eqref{eq:bjbk1} and simplifying gives
\begin{align}\label{eq:bjbk2} 
& (\bb_j-\bb_k)\wp_2(\ta_j-\ta_k) \nonumber \\
&=    -\frac12(r_k\bs_k+r_j\bs_j)\wp_2'(\ta_j-\ta_k)  -\frac12(r_k\bs_k+r_j\bs_j)f_2'(\ta_j-\ta_k)            \nonumber \\
& \phantom{=\;}-\sum_{l\neq j,k}^{\cN} r_l \bs_l \big(\zeta_2(\ta_j-\ta_k)-\zeta_2(\ta_j-\ta_l)\big)\wp_2(\ta_k-\ta_l) \nonumber \\
& \phantom{=\;}-\frac12\sum_{l\neq j,k}^{\cN} r_l \bs_l f_2'(\ta_j-\ta_k)-\frac12\sum_{l\neq j,k}^{\cN} r_l\bs_l f_2'(\ta_k-\ta_l)  \nonumber \\
&=  -\frac12(r_k\bs_k+r_j\bs_j)\wp_2'(\ta_j-\ta_k)-\sum_{l\neq j,k}^{\cN} r_l \bs_l \big(\zeta_2(\ta_j-\ta_k)-\zeta_2(\ta_j-\ta_l)\big)\wp_2(\ta_k-\ta_l) \nonumber \\
& \phantom{=\;}-\frac12\sum_{l=1}^{\cN} r_l \bs_l f_2'(\ta_j-\ta_k)-\frac12\sum_{l\neq j,k}^{\cN} r_l \bs_l f_2'(\ta_k-\ta_l)  \nonumber \\
&=   -\frac12(r_k\bs_k+r_j\bs_j)\wp_2'(\ta_j-\ta_k) -\sum_{l\neq j,k}^{\cN} r_l \bs_l \big(\zeta_2(\ta_j-\ta_k)-\zeta_2(\ta_j-\ta_l)\big)\wp_2(\ta_k-\ta_l) \nonumber \\
& \phantom{=\;}-\frac12\sum_{l\neq j,k}^{\cN} r_l\bs_l f_2'(\ta_k-\ta_l),
\end{align}
where we have used Lemma~\ref{lem:totalspin} in the final step to replace $\sum_{l=1}^{\cN}r_l\bs_l$ by $\boldsymbol{0}$. 

Inserting \eqref{eq:bjbk2} into \eqref{eq:Qjdot2} gives
\begin{align}\label{eq:Qjdot3}
\dot{Q}_j=&\; \frac12 \sum_{k\neq j}^{\cN} (r_k\bs_k+r_j\bs_j)\cdot (\bs_j\wedge\bs_k)\wp_2'(\ta_j-\ta_k) \nonumber \\
&\; +\sum_{k\neq j}^{\cN}\sum_{l\neq j,k}^{\cN} r_l \bs_l\cdot(\bs_j\wedge\bs_k)\big(\zeta_2(\ta_j-\ta_k)-\zeta_2(\ta_j-\ta_l)\wp_2(\ta_k-\ta_l)           \nonumber \\
&\; +\frac12\sum_{k\neq j}^{\cN} \sum_{l\neq j,k}^{\cN} r_l \bs_l\cdot(\bs_j\wedge\bs_k)f_2'(\ta_j-\ta_k)     \nonumber \\
&\; +\ii\bs_j\cdot\dbphi-\sum_{k\neq j}^{\cN}\sum_{l\neq k}^{\cN} (r_k+r_l) \bs_j\cdot(\bs_k\wedge\bs_l)\zeta_2(\ta_j-\ta_k)\wp_2(\ta_k-\ta_l).
\end{align}
The sum in the first line of \eqref{eq:Qjdot3} vanishes as a consequence of \eqref{eq:cyclic}. The double sum in the second line of \eqref{eq:Qjdot3} may be symmetrized, 
\begin{multline}\label{eq:Qjdotsum1}
\sum_{k\neq j}^{\cN}\sum_{l\neq j,k}^{\cN} r_l \bs_l\cdot(\bs_j\wedge\bs_k)\big(\zeta_2(\ta_j-\ta_k)-\zeta_2(\ta_j-\ta_l)\big)\wp_2(\ta_k-\ta_l)   \\
=\frac12\sum_{k\neq j}^{\cN}\sum_{l\neq j,k}^{\cN} (r_k+r_l) \bs_l\cdot(\bs_j\wedge\bs_k)\big(\zeta_2(\ta_j-\ta_k)-\zeta_2(\ta_j-\ta_l)\big)\wp_2(\ta_k-\ta_l),
\end{multline}
(using \eqref{eq:cyclic}, the antisymmetry of $\wedge$, and the fact that $\wp(z)$ is an even function \eqref{eq:parity}) and the final sum in \eqref{eq:Qjdot3} may be rewritten as
\begin{multline}\label{eq:Qjdotsum2}
\sum_{k\neq j}^{\cN}\sum_{l\neq k}^{\cN} (r_k+r_l) \bs_j\cdot(\bs_k\wedge\bs_l)\zeta_2(\ta_j-\ta_k)\wp_2(\ta_k-\ta_l) \\
= \sum_{k\neq j}^{\cN}(r_k+r_j)\bs_j\cdot(\bs_k\wedge\bs_j)\zeta_2(\ta_j-\ta_k)\wp_2(\ta_k-\ta_l) +\sum_{k\neq j}^{\cN}\sum_{l\neq k}^{\cN} (r_k+r_l) \bs_j\cdot(\bs_k\wedge\bs_l)\zeta_2(\ta_j-\ta_k)\wp_2(\ta_k-\ta_l) \\
=\frac12\sum_{k\neq j}^{\cN}\sum_{l\neq k}^{\cN} (r_k+r_l) \bs_j\cdot(\bs_k\wedge\bs_l)\big(\zeta_2(\ta_j-\ta_k)-\zeta_2(\ta_j-\ta_l)\big)\wp_2(\ta_k-\ta_l),
\end{multline}
where we have used $\bs_j\cdot(\bs_k\wedge\bs_j)=0$ and, similarly as in \eqref{eq:Qjdotsum1}, symmetrized the double sum in the final step. Hence, using \eqref{eq:cyclic}, we see that \eqref{eq:Qjdotsum1} and \eqref{eq:Qjdotsum2} are equal, leading to cancellation in \eqref{eq:Qjdot3}. We are left with
\begin{align}\label{eq:Qjdot4}
\dot{Q}_{j}=&\; \frac12\sum_{k\neq j}^{\cN}\sum_{l\neq j,k}^{\cN} r_l\bs_l\cdot(\bs_j\wedge\bs_k)f_2'(\ta_k-\ta_l)+\ii \bs_j\cdot\dbphi \nonumber\\
 =&\; \frac14\sum_{k\neq j}^{\cN}\sum_{l\neq j,k}^{\cN} (r_k+r_l)\bs_l\cdot(\bs_j\wedge\bs_k)f_2'(\ta_k-\ta_l)-\frac14\sum_{k=1}^{\cN}\sum_{l\neq k}^{\cN} (r_k+r_l)\bs_j\cdot(\bs_k\wedge\bs_l)f_2'(\ta_k-\ta_k),
\end{align}
where we have symmetrized the double sum (using \eqref{eq:cyclic}, the antisymmetry of $\wedge$, and the fact that $f_2'(z)$ is an odd function \eqref{eq:parity}) and inserted \eqref{eq:phidot} in the second step. Noting that all terms proportional to $\bs_j\cdot(\bs_k\wedge\bs_l)$ with $j=k$ and $j=l$ are zero in the second double sum in \eqref{eq:Qjdot4} and using \eqref{eq:cyclic}, we see that $\dot{Q}_j=0$. 

\subsubsection{Conservation of $R$}

We write
\begin{equation}\label{eq:Rsum}
R=R^{(1)}+R^{(2)},
\end{equation}
where
\begin{equation}\label{eq:R1R2}
R^{(1)}\coloneqq \bphi^2,\qquad R^{(2)}\coloneqq -\rho^2-\frac12\sum_{j=1}^{\cN}\sum_{k\neq j}^{\cN} r_jr_k\bs_j\cdot\bs_k f_2(\ta_j-\ta_k). 
\end{equation}

By differentiating $R^{(1)}$ with respect to $t$ and inserting \eqref{eq:phidotSH} and 
\begin{equation}
\bphi=-\ii \bb_j-\ii\sum_{k\neq j}^{\cN}r_k \bs_k \zeta_2(\ta_j-\ta_k) \quad (j=1,\ldots,\cN),
\end{equation}
which follows from \eqref{eq:bj}, we compute
\begin{align}\label{eq:R1dot}
\dot{R}^{(1)}= 2\bphi\cdot\dbphi=&\; \frac{\ii}{2}\sum_{j=1}^{\cN}\sum_{k\neq j}^{\cN} (r_j+r_k)\bphi\cdot(\bs_j\wedge\bs_k )f_2'(\ta_j-\ta_k) \nonumber \\
=&\; \frac{\ii}{2}\sum_{j=1}^{\cN}\sum_{k\neq j}^{\cN} r_j\Bigg(-\ii \bb_k-\ii \sum_{l\neq j}^{\cN} r_l\bs_l \zeta_2(\ta_k-\ta_l)\Bigg)\cdot(\bs_j\wedge\bs_k)f_2'(\ta_j-\ta_k)  \nonumber \\
&\; + \frac{\ii}{2}\sum_{j=1}^{\cN}\sum_{k\neq j}^{\cN} r_k\Bigg(-\ii \bb_j-\ii \sum_{l\neq k}^{\cN} r_l\bs_l \zeta_2(\ta_j-\ta_l)\Bigg)\cdot(\bs_j\wedge\bs_k)f_2'(\ta_j-\ta_k) \nonumber \\
=&\; \frac12 \sum_{j=1}^{\cN}\sum_{k\neq j}^{\cN} (r_j\bb_k+r_k\bb_j)\cdot(\bs_j\wedge\bs_k)f_2'(\ta_j-\ta_k) \nonumber\\
&\; +\frac12\sum_{j=1}^{\cN}\sum_{k\neq j}^{\cN}  \sum_{l\neq j,k}^{\cN} \bs_j \cdot(\bs_k\wedge\bs_l)\big(r_jr_l\zeta_2(\ta_k-\ta_l)+r_kr_l\zeta_2(\ta_j-\ta_l)\big)f_2'(\ta_j-\ta_k),
\end{align}
where we have used \eqref{eq:cyclic} in the last step. 

By differentiating $R^{(2)}$ with respect to $t$, using $\dot{\tilde{a}}_j=\dot{a}_j$ for $j=1,\ldots,\cN$, and inserting \eqref{eq:sjdotSH} and \eqref{eq:ajdotSH} in the form
\begin{equation}
\dot{a}_j\bs_j=-r_j\bs_j\wedge\bb_j \quad (j=1,\ldots,\cN),
\end{equation}
we compute
\begin{align}\label{eq:R2dot}
\dot{R}^{(2)}=&\; -\frac12 \sum_{j=1}^{\cN}\sum_{k\neq j}^{\cN} r_jr_k\big((\dot{\bs}_j\cdot\bs_k+\bs_j\cdot\dot{\bs}_k)f_2(\ta_j-\ta_k)+\bs_j\cdot\bs_k(\dot{a}_j-\dot{a}_k)f_2'(a_j-a_k)\big) \nonumber \\
=&\; \frac12 \sum_{j=1}^{\cN}\sum_{k\neq j}^{\cN} \Bigg\{ \sum_{l\neq j}^{\cN} (r_jr_k+r_kr_l) (\bs_j\wedge\bs_l)\cdot\bs_k \wp_2(\ta_j-\ta_l) \nonumber \\
&\;\phantom{\frac12 \sum_{j=1}^{\cN}\sum_{k\neq j}^{\cN} \Bigg(}+\sum_{l\neq k}^{\cN} (r_jr_k+r_jr_l)\bs_j\cdot(\bs_k\wedge\bs_l)\wp_2(\ta_k-\ta_l)\Bigg\} f_2(\ta_j-\ta_k) \nonumber \\
&\; +\frac12 \sum_{j=1}^{\cN} \sum_{k\neq j}^{\cN} \big(r_k(\bs_j\wedge\bb_j)\cdot\bs_k)-r_j\bs_j\cdot(\bs_k\wedge\bb_k)\big) f_2'(\ta_j-\ta_k) \nonumber \\
=&\; - \frac12\sum_{j=1}^{\cN}\sum_{k\neq j}^{\cN}\sum_{l\neq j,k}^{\cN} (\bs_j\wedge\bs_k)\cdot\bs_l \big((r_jr_k+r_kr_l)\wp_2(\ta_j-\ta_l)-(r_jr_k+r_jr_l)\wp_2(\ta_k-\ta_l)\big)f_2(\ta_j-\ta_k) \nonumber \\
&\; -\frac12\sum_{j=1}^{\cN}\sum_{k\neq j}^{\cN} (\bs_j\wedge\bs_k)\cdot(r_j\bb_j+r_j\bb_k)f_2'(\ta_j-\ta_k),
\end{align}
where we have again used \eqref{eq:cyclic} in the last step. Hence, differentiating \eqref{eq:Rsum} with respect to $t$ and inserting \eqref{eq:R1dot} and \eqref{eq:R2dot}, we see that the terms in $\bb_j$, $\bb_k$ cancel and we are left with
\begin{align}\label{eq:Rdot1}
\dot{R}=\dot{R}^{(1)}+\dot{R}^{(2)}=&\; \frac12\sum_{j=1}^{\cN}\sum_{k\neq j}^{\cN}  \sum_{l\neq j,k}^{\cN}  \bs_j \cdot(\bs_k\wedge\bs_l)\big(r_jr_l\zeta_2(\ta_k-\ta_l)+r_kr_l\zeta_2(\ta_j-\ta_l)\big)f_2'(\ta_j-\ta_k) \nonumber \\
&\; - \frac12\sum_{j=1}^{\cN}\sum_{k\neq j}^{\cN}\sum_{l\neq j,k}^{\cN} (\bs_j\wedge\bs_k)\cdot\bs_l \big\{(r_jr_k+r_kr_l)\wp_2(\ta_j-\ta_l) \nonumber \\
&\; \phantom{- \frac12\sum_{j=1}^{\cN}\sum_{k\neq j}^{\cN}\sum_{l\neq j,k}^{\cN} (\bs_j\wedge\bs_k)\cdot\bs_l \big(}-(r_jr_k+r_jr_l)\wp_2(\ta_k-\ta_l)\big\}f_2(\ta_j-\ta_k) \nonumber \\
=&\; \frac12\sum_{j=1}^{\cN}\sum_{k\neq j}^{\cN}  \sum_{l\neq j,k}^{\cN}  \bs_j \cdot(\bs_k\wedge\bs_l) \big\{-r_jr_l\partial_{\ta_k}\big( \zeta_2(\ta_k-\ta_l)f_2(\ta_j-\ta_k)   \big) \nonumber \\
&\; \phantom{\frac12\sum_{j=1}^{\cN}\sum_{k\neq j}^{\cN}  \sum_{l\neq j,k}^{\cN}  \bs_j \cdot(\bs_k\wedge\bs_l) \big\{} +r_kr_l \partial_{\ta_j}\big(\zeta_2(\ta_j-\ta_l)f_2(\ta_j-\ta_k)\big) \nonumber \\
&\;  \phantom{\frac12\sum_{j=1}^{\cN}\sum_{k\neq j}^{\cN}\sum_{l\neq j,k}^{\cN} \bs_j \cdot(\bs_k\wedge\bs_l) \big\{} -r_jr_k\partial_{\ta_l}\big(\zeta_2(\ta_j-\ta_l)-\zeta_2(\ta_k-\ta_l)\big)f_2(\ta_j-\ta_k)\big\}.
\end{align}
By permuting indices $k\leftrightarrow l$ and $j\leftrightarrow l$ in the first and second terms in the summand, respectively, we find that
\begin{equation}\label{eq:Rdot2}
\dot{R}= -\frac12 \sum_{j=1}^{\cN}\sum_{k\neq j}^{\cN}  \sum_{l\neq j,k}^{\cN} r_jr_k \bs_j\cdot(\bs_k\wedge\bs_l)\partial_{\ta_l} g(\ta_j,\ta_k,\ta_l),
\end{equation}
where
\begin{align}\label{eq:g}
g(\ta_j,\ta_k,\ta_l)\coloneqq &\; \zeta_2(\ta_k-\ta_l)f_2(\ta_j-\ta_l)-\zeta_2(\ta_j-\ta_l)f_2(\ta_k-\ta_l) \nonumber \\
&\; +\big(\zeta_2(\ta_j-\ta_l)-\zeta_2(\ta_k-\ta_l)\big)f_2(\ta_j-\ta_k).
\end{align}
The function $g(\ta_j,\ta_k,\ta_l)$ is a meromorphic function of $\ta_l$ with no poles in the parallelogram $\Pi$ defined by vertices at $(0,0)$, $(2\ell,0)$, $(0,-2\ii\delta)$, and $(2\ell,-2\ii\delta)$. It follows from \eqref{eq:realperiod}--\eqref{eq:imperiod} that $g(\ta_j,\ta_k,\ta_l+2\ii\delta)=g(\ta_j,\ta_j,\ta_l)$ and 
\begin{equation}\label{eq:gperiod}
g(\ta_j,\ta_k,\ta_l+2\ell)=g(\ta_j,\ta_k,\ta_l)-\frac{\pi}{\delta}\big(f_2(\ta_j-\ta_l)-f_2(\ta_k-\ta_l)\big)+\bigg(\frac{\pi}{\delta}\bigg)^2\big(\zeta_2(\ta_j-\ta_l)-\zeta_2(\ta_k-\ta_l)\big).
\end{equation} 

By adding certain terms to $g(\ta_j,\ta_k,\ta_l)$, we obtain a function $\tilde{g}(\ta_j,\ta_k,\ta_l)$ which is doubly-periodic with respect to $\ta_l$ and has no poles $\ta_l\in \Pi$ and thus, by Liouville's theorem, is a constant function of $\ta_l$. Let
\begin{equation}\label{eq:gt}
\tilde{g}(\ta_j,\ta_k,\ta_l)\coloneqq g(\ta_j,\ta_k,\ta_k)-h(\ta_j-\ta_l)+h(\ta_k-\ta_l)
\end{equation}
where
\begin{equation}\label{eq:h}
h(z)\coloneqq \zeta_2(z)\big(f_2(z)-f_2(0)\big)-\frac23\bigg(\zeta_2(z)^3+3\frac{\zeta(\ii\delta)}{\ii\delta}\zeta_2(z)+\frac12\wp_2'(z)\bigg).
\end{equation}
The function $h(z)$ is seen to be regular at $z=0$ using the Laurent series
\begin{equation}\label{eq:laurent}
\zeta_2(z)=\frac{1}{z}-\frac{\zeta(\ii\delta)}{\ii\delta} z+O(z^3),\qquad \wp_2(z)=\frac{1}{z^2}-\frac{\zeta(\ii\delta)}{\ii\delta}+O(z^2)
\end{equation}
as $z\to 0$; the series in \eqref{eq:laurent} follow from those for $\zeta(z)$ and $\wp(z)$ \cite[Chapter~23.9]{DLMF} and the definitions of $\zeta_2(z)$ \eqref{eq:zeta2} and $\wp_2(z)$ \eqref{eq:wp2}. Hence the functions $h(\ta_j-\ta_l)$ and $h(\ta_k-\ta_l)$ are regular for $\ta_l\in\Pi$. Moreover, $h(z)$ is $2\ii\delta$-periodic and satisfies the identity
\begin{align}\label{eq:hperiod}
h(\ta_j-\ta_l-2\ell)-h(\ta_k-\ta_l-2\ell)=&\; h(\ta_j-\ta_k)-h(\ta_k-\ta_l)-\frac{\pi}{\delta}\big(f_2(\ta_j-\ta_l)-f_2(\ta_k-\ta_l)\big) \nonumber \\
&\;+\bigg(\frac{\pi}{\delta}\bigg)^2\big(\zeta_2(\ta_j-\ta_l)-\zeta_2(\ta_k-\ta_l)\big),
\end{align}
by \eqref{eq:h} with \eqref{eq:realperiod}--\eqref{eq:imperiod}.

The function $\tilde{g}(\ta_j,\ta_k,\ta_l)$ \eqref{eq:gt} is thus analytic in $\ta_l$ when $\ta_l\in\Pi$. Using \eqref{eq:gperiod} and \eqref{eq:hperiod}, it follows that $\tilde{g}(\ta_j,\ta_k,\ta_l+2\ell)=\tilde{g}(\ta_j,\ta_k,\ta_l+2\ii\delta)=\tilde{g}(\ta_j,\ta_k,\ta_l)$. Hence $\tilde{g}(\ta_j,\ta_k,\ta_l)$ is constant with respect to $\ta_l$. 

Inserting \eqref{eq:gt} into \eqref{eq:Rdot2} gives 
\begin{align}
\dot{R}=&\; -\frac12 \sum_{j=1}^{\cN}\sum_{k\neq j}^{\cN}  \sum_{l\neq j,k}^{\cN} r_jr_k \bs_j\cdot(\bs_k\wedge\bs_l)\partial_{\ta_l}\tilde{g}(\ta_j,\ta_k,\ta_l) \nonumber \\
&\; -\frac{1}{2}\sum_{j=1}^{\cN}\sum_{k\neq j}^{\cN}  \sum_{l\neq j,k}^{\cN} r_jr_k \bs_j\cdot(\bs_k\wedge\bs_l)\partial_{\ta_l}\big(h(\ta_j-\ta_l)-h(\ta_k-\ta_l)   \big).
\end{align}
The sum in the first line vanishes because $\partial_{\ta_k}\tilde{g}(\ta_j,\ta_k,\ta_l)=0$. The second sum vanishes by \eqref{eq:constraint3SH} and \eqref{eq:constraint1SH}. We conclude that $\dot{R}=0$. 

\section{B\"{a}cklund transformation}\label{sec:Backlund}

We prove the Bäcklund transformation between the elliptic spin CM systems \eqref{eq:sCM1} and \eqref{eq:sCM2} stated in Theorem~\ref{thm:backlund} in Section~\ref{subsec:Backlundproof}. Building on this result and using the results of Sections~\ref{sec:solitons} and \ref{sec:conserved}, we prove Theorem~\ref{thm:main} in Section~\ref{subsec:mainproof}.

\subsection{Proof of Theorem~\ref{thm:backlund}}\label{subsec:Backlundproof}

This proof consists in deriving the deriving second-order equations
\begin{equation}\label{eq:ajddotSH}
\ddot{a}_j=-\sum_{k\neq j}^{\cN}(1+r_jr_k) \bs_j\cdot\bs_k \wp_2'(a_j-a_k) \quad (j=1,\ldots,\cN),
\end{equation}
equivalent to \eqref{eq:sCMa}, \eqref{eq:sCMb}, via the notation \eqref{eq:shorthand}, as a consequence of the first-order equations \eqref{eq:phidot}, \eqref{eq:sCMs}, \eqref{eq:sCMt}, and \eqref{eq:ajdot} in the form \eqref{eq:phidotSH}, \eqref{eq:sjdotSH}, and \eqref{eq:ajdotSH}. We use that the constraints \eqref{eq:constraint1}--\eqref{eq:constraint3} hold on $[0,T)$ by Proposition~\ref{prop:conserved} and Lemma~\ref{lem:totalspin}.

Recalling the definition of $\bb_j$ \eqref{eq:bj}, \eqref{eq:ajdotSH} can be written as
\begin{equation}\label{eq:ajdotSH2}
\dot{a}_j \bs_j= - r_j \bs_j \wedge \bb_j  \quad (j=1,\ldots,\cN).
\end{equation}
Differentiating \eqref{eq:ajdotSH2} with respect to $t$ and rearranging gives
\begin{equation}\label{eq:ajddot1}
\ddot{a}_j \bs_j = -\dot{a}_j\dot{\bs}_j - r_j \dot{\bs}_j\wedge\bb_j-r_j\bs_j\wedge\dbb_j. 
\end{equation}
We compute the terms on the right hand side of \eqref{eq:ajddot1}. Using \eqref{eq:sjdotSH} and then \eqref{eq:ajdotSH},  
\begin{align}\label{eq:ajddot2}
-\dot{a}_j\dot{\bs}_j-r_j \dot{\bs}_j\wedge\bb_j= &\; \sum_{k\neq j}^{\cN} (1+r_jr_k) \dot{a}_j \bs_j \wedge\bs_k \wp_2(a_j-a_k) + \sum_{k\neq j}^{\cN} r_j (1+r_j r_k) (\bs_j\wedge\bs_k)\wedge\bb_j \wp_2(a_j-a_k) \nonumber \\
=&\; -\sum_{k\neq j}^{\cN} r_j(1+r_jr_k)(\bs_j\wedge\bb_j)\wedge\bs_k \wp_2(a_j-a_k) \nonumber \\
&\; +\sum_{k\neq j}^{\cN}r_j (1+ r_jr_k) (\bs_j\wedge\bs_k)\wedge\bb_j \wp_2(a_j-a_k),
\end{align}
To simplify, we use $r_j^2=1$, the standard vector identities
\begin{equation}\label{eq:VecId}
(\bx\wedge\by)\wedge\bz=-(\by\cdot\bz)\bx+(\bx\cdot\bz)\by,\qquad \bx\wedge(\by\wedge\bz)=(\bx\cdot\bz)\by-(\bx\cdot\by)\bz,
\end{equation}
and \eqref{eq:constraint2} in the form $\bb_j\cdot\bs_j=0$. Hence, 
\begin{equation}\label{eq:ajddot4}
-\dot{a}_j\dot{\bs}_j-r_j \dot{\bs}_j\wedge\bb_j =- \sum_{k\neq j}^{\cN} (r_j+r_k) (\bs_j\cdot\bs_k)\bb_j \wp_2(a_j-a_k).
\end{equation}

To compute the remaining term in \eqref{eq:ajddot1}, we first differentiate \eqref{eq:bj} with respect to $t$ to find
\begin{equation}\label{eq:bjdot1}
\dbb_j= \ii \dbphi-\sum_{k\neq j}^{\cN} r_k\dot{\bs}_k\zeta_2(\ta_j-\ta_k)+\sum_{k\neq j}^{\cN} r_k\bs_k \wp_2(\ta_j-\ta_k)(\dot{a}_j-\dot{a}_k) 
\end{equation}
where we have used that $\dot{\tilde{a}}_j=\dot{a}_j$. Taking the cross product with $-r_j\bs_j$ and using \eqref{eq:sjdotSH} gives
\begin{align}\label{eq:bjdot2}
-r_j\bs_j\wedge \dbb_j=&\;  -\ii r_j \bs_j\wedge\dbphi+\sum_{k\neq j}^{\cN} r_j r_k\bs_j\wedge\dot{\bs}_k\zeta_2(\ta_j-\ta_k)-\sum_{k\neq j}^{\cN} r_jr_k\bs_j\wedge\bs_k \wp_2(\ta_j-\ta_k)(\dot{a}_j-\dot{a}_k) \nonumber \\
=&\;  -\ii r_j \bs_j\wedge\dbphi-\sum_{k\neq j}^{\cN}\sum_{l\neq k}^{\cN} r_j (r_k+r_l)\bs_j\wedge(\bs_k\wedge\bs_l)\zeta_2(\ta_j-\ta_k)\wp_2(\ta_k-\ta_l) \nonumber \\
&\; -\sum_{k\neq j}^{\cN}  r_jr_k\big((\dot{a}_j\bs_j)\wedge\bs_k-\bs_j\wedge(\dot{a}_k\bs_k)\big) \wp_2(\ta_j-\ta_k).
\end{align}
The double sum in \eqref{eq:bjdot2} can be rewritten as
\begin{multline}\label{eq:firstsum1}
\sum_{k\neq j}^{\cN}\sum_{l\neq k}^{\cN} r_j (r_k+r_l)\bs_j\wedge(\bs_k\wedge\bs_l)\zeta_2(\ta_j-\ta_k)\wp_2(\ta_k-\ta_l) \\
= \sum_{k\neq j}^{\cN} r_j(r_k+r_j) \bs_j\wedge(\bs_k\wedge\bs_j)\zeta_2(\ta_j-\ta_k)\wp_2(\ta_k-\ta_j) +\sum_{k\neq j}^{\cN}\sum_{l\neq j,k}^{\cN} r_j (r_k+r_l)\bs_j\wedge(\bs_k\wedge\bs_l)\zeta_2(\ta_j-\ta_k)\wp_2(\ta_k-\ta_l) \\
\begin{aligned}= &\; \sum_{k\neq j}^{\cN} (1+r_jr_k) \bs_j\wedge(\bs_k\wedge\bs_j)\zeta_2(\ta_j-\ta_k)\wp_2(\ta_j-\ta_k)\\
&\; +\sum_{k\neq j}^{\cN}\sum_{l\neq j,k}^{\cN} r_j (r_k+r_l)\bs_j\wedge(\bs_k\wedge\bs_l)\big(\zeta_2(\ta_j-\ta_k-\zeta_2(\ta_j-\ta_l)\wp_2(\ta_k-\ta_l),
\end{aligned}
\end{multline}
using $r_j^2=1$ and the parity properties of $\wp_2(z)$ and $\zeta_2(z)$ \eqref{eq:parity} in the second step. Then, the second identity in \eqref{eq:VecId} and the constraint \eqref{eq:constraint1} yield
\begin{multline}\label{eq:firstsum2}
\sum_{k\neq j}^{\cN}\sum_{l\neq k}^{\cN} r_j (r_k+r_l)\bs_j\wedge(\bs_k\wedge\bs_l)\zeta_2(\ta_j-\ta_k)\wp_2(\ta_k-\ta_l) \\
 \begin{aligned}= &\; -\sum_{k\neq j}^{\cN} (1+r_jr_k) (\bs_j\cdot\bs_k)\bs_j \zeta_2(\ta_j-\ta_k)\wp_2(\ta_k-\ta_j)\\
&\; +\sum_{k\neq j}^{\cN}\sum_{l\neq j,k}^{\cN} r_j (r_k+r_l)\big((\bs_j\cdot\bs_l)\bs_k-(\bs_j\cdot\bs_k)\bs_l  \big)\big(\zeta_2(\ta_j-\ta_k)-\zeta_2(\ta_j-\ta_l)\big)\wp_2(\ta_k-\ta_l).
\end{aligned}
\end{multline}
We simplify the remaining sum in \eqref{eq:bjdot2} using \eqref{eq:ajdotSH}, $r_j^2=r_k^2=1$, and \eqref{eq:VecId}:
\begin{align}\label{eq:secondsum}
r_jr_k\big((\dot{a}_j\bs_j)\wedge\bs_k-\bs_j\wedge(\dot{a}_k\bs_k)\big)=&\; -r_k(\bs_j\wedge\bb_j)\wedge\bs_k+r_j\bs_j\wedge(\bs_k\wedge\bb_k) \nonumber \\
=&\;  r_k(\bb_j\cdot\bs_k)\bs_j-r_k(\bs_j\cdot\bs_k)\bb_j +r_j(\bs_j\cdot\bb_k)\bs_k-r_j(\bs_j\cdot\bs_k)\bb_k.
\end{align}
By using \eqref{eq:firstsum2} and \eqref{eq:secondsum} in \eqref{eq:bjdot2}, we arrive at
\begin{align}\label{eq:bjdot3}
-r_j\bs_j\wedge \dbb_j=&\;  -\ii r_j \bs_j\wedge\dbphi+\sum_{k\neq j}^{\cN} (1+r_jr_k) (\bs_j\cdot\bs_k)\bs_j \zeta_2(\ta_j-\ta_k)\wp_2(\ta_k-\ta_j) \nonumber \\
&\; -\sum_{k\neq j}^{\cN}\sum_{l\neq j,k}^{\cN} r_j (r_k+r_l)\big((\bs_j\cdot\bs_l)\bs_k-(\bs_j\cdot\bs_k)\bs_l  \big)\big(\zeta_2(\ta_j-\ta_k)-\zeta_2(\ta_j-\ta_l)\big)\wp_2(\ta_k-\ta_l) \nonumber \\
&\; -\sum_{k\neq j}^{\cN} \big(r_k(\bb_j\cdot\bs_k)\bs_j-r_k(\bs_j\cdot\bs_k)\bb_j +r_j(\bs_j\cdot\bb_k)\bs_k-r_j(\bs_j\cdot\bs_k)\bb_k \big) \wp_2(\ta_j-\ta_k),
\end{align}
Then inserting \eqref{eq:ajddot4} and \eqref{eq:bjdot3} into \eqref{eq:ajddot1}, we get
\begin{align}\label{eq:ajddot5}
\ddot{a}_j \bs_j =&\; -\ii r_j \bs_j\wedge\dbphi+\sum_{k\neq j}^{\cN} (1+r_jr_k) (\bs_j\cdot\bs_k)\bs_j \zeta_2(\ta_j-\ta_k)\wp_2(\ta_k-\ta_j) \nonumber \\
&\; -\sum_{k\neq j}^{\cN}\sum_{l\neq j,k}^{\cN} r_j (r_k+r_l)\big((\bs_j\cdot\bs_l)\bs_k-(\bs_j\cdot\bs_k)\bs_l  \big)\big(\zeta_2(\ta_j-\ta_k)-\zeta_2(\ta_j-\ta_l)\big)\wp_2(\ta_k-\ta_l) \nonumber \\
&\; -\sum_{k\neq j}^{\cN} \big(r_k(\bb_j\cdot\bs_k)\bs_j+r_j(\bs_j\cdot\bb_k)\bs_k-r_j(\bs_j\cdot\bs_k)\bb_k \big) \wp_2(\ta_j-\ta_k),\nonumber \\
=&\; -\ii r_j \bs_j\wedge\dbphi+\sum_{k\neq j}^{\cN} (1+r_jr_k) (\bs_j\cdot\bs_k)\bs_j \zeta_2(\ta_j-\ta_k)\wp_2(\ta_k-\ta_j) \nonumber \\
&\; -\sum_{k\neq j}^{\cN}\sum_{l\neq j,k}^{\cN} r_j (r_k+r_l)\big((\bs_j\cdot\bs_l)\bs_k-(\bs_j\cdot\bs_k)\bs_l  \big)\big(\zeta_2(\ta_j-\ta_k)-\zeta_2(\ta_j-\ta_l)\big)\wp_2(\ta_k-\ta_l) \nonumber \\
&\; -\sum_{k\neq j}^{\cN} \big(r_k((\bb_j-\bb_k)\cdot\bs_k)\bs_j-r_j(\bs_j\cdot(\bb_j-\bb_k))\bs_k+r_j(\bs_j\cdot\bs_k)(\bb_j-\bb_k) \big) \wp_2(\ta_j-\ta_k),
\end{align}
using \eqref{eq:constraint2} in the form $\bs_j\cdot\bb_j=\bs_k\cdot\bb_k=0$ in the second step.

The remainder of the proof consists of using known expressions for $(\bb_j-\bb_k)\wp_2(\ta_j-\ta_k)$ \eqref{eq:bjbk2} and $\dbphi$ \eqref{eq:phidotSH} and elliptic function identities to show that \eqref{eq:ajddot5} becomes \eqref{eq:ajddotSH}.

We insert \eqref{eq:bjbk2} (which was derived under the assumption \eqref{eq:constraint3}) into \eqref{eq:ajddot5} to obtain, after combining some terms,
\begin{align}\label{eq:ajddot6}
\ddot{a}_j\bs_j =&\;  -\ii r_j \bs_j\wedge\dbphi+\sum_{k\neq j}^{\cN} (1+r_jr_k) (\bs_j\cdot\bs_k)\bs_j \zeta_2(\ta_j-\ta_k)\wp_2(\ta_j-\ta_k) \nonumber \\
&\;  -\frac12 \sum_{k\neq j}^{\cN} (1+r_jr_k) (\bs_j\cdot\bs_k)\bs_j \wp_2'(\ta_j-\ta_k)    \nonumber \\
&\;  +\sum_{k\neq j}^{\cN}\sum_{l\neq j,k}^{\cN} r_k r_l (\bs_l\cdot\bs_k)\bs_j \big(\zeta_2(\ta_j-\ta_k)-\zeta_2(\ta_j-\ta_l)\big)\wp_2(\ta_k-\ta_l) \nonumber \\
&\; +\frac12 \sum_{k\neq j}^{\cN} \sum_{l\neq j,k}^{\cN}r_k  r_l(\bs_l \cdot\bs_k)\bs_j f_2'(\ta_k-\ta_l)   -\sum_{k\neq j}^{\cN}\sum_{l\neq j,k}^{\cN} r_jr_l \big((\bs_j\cdot \bs_l)\bs_k-(\bs_j\cdot\bs_k)\bs_l  \big) f_2'(\ta_k-\ta_l).
\end{align}
Since $\wp(z)$ is an even function \eqref{eq:parity}, the double sum in the third line of \eqref{eq:ajddot6} is antisymmetric under the interchange of $k$ and $l$ and hence vanishes. The first double sum in the fourth line of \eqref{eq:ajddot6} similarly vanishes by symmetry, because $f_2'(z)$ is an odd function \eqref{eq:parity}. We again use the identity \eqref{eq:EllipticId1}, leading to, after some rearrangement,
\begin{align}\label{eq:ajddot7}
\ddot{a}_j\bs_j =&\;  -\ii r_j \bs_j\wedge\dbphi-\sum_{k\neq j}^{\cN} (1+r_jr_k) (\bs_j\cdot\bs_k)\bs_j \wp_2'(\ta_j-\ta_k)  -\frac12 \sum_{k\neq j}^{\cN} (1+r_jr_k)(\bs_j\cdot\bs_k)\bs_j f_2'(\ta_j-\ta_k) \nonumber \\
&\; -\frac12 \sum_{k\neq j}^{\cN}\sum_{l\neq j,k}^{\cN} r_j r_l \big((\bs_j\cdot\bs_l)\bs_k-(\bs_j\cdot\bs_k)\bs_l\big)f_2'(\ta_k-\ta_l).
\end{align}
The second identity in \eqref{eq:VecId} and
\begin{multline}
 \sum_{k\neq j}^{\cN}\sum_{l\neq j,k}^{\cN} r_j r_l \big((\bs_j\cdot\bs_l)\bs_k-(\bs_j\cdot\bs_k)\bs_l\big)f_2'(\ta_k-\ta_l) \\
 =\sum_{k=1}^{\cN}\sum_{l\neq k}^{\cN} r_j r_l \big((\bs_j\cdot\bs_l)\bs_k-(\bs_j\cdot\bs_k)\bs_l\big)f_2'(\ta_k-\ta_l)-\sum_{k\neq j}^{\cN} (1+r_jr_k)(\bs_j\cdot\bs_k)\bs_j f_2'(\ta_j-\ta_k)
\end{multline}
lead to 
\begin{align}\label{eq:ajddot8}
\ddot{a}_j\bs_j =&\;  -\ii r_j \bs_j\wedge\dbphi-\sum_{k\neq j}^{\cN} (1+r_jr_k) (\bs_j\cdot\bs_k)\bs_j \wp_2'(\ta_j-\ta_k) -\frac12 \sum_{k=1}^{\cN}\sum_{l\neq k}^{\cN} r_j r_l \bs_j\wedge(\bs_k\wedge\bs_l)    f_2'(\ta_k-\ta_l).
\end{align}
Symmetrizing the double sum (using the antisymmetry of $\wedge$ and the fact that $f_2'(z)$ is an odd function \eqref{eq:parity}) and inserting \eqref{eq:phidotSH} gives the result \eqref{eq:ajddotSH} after recalling that $\ta_j-\ta_k=a_j-a_k$ for $r_j=r_k$. 

\subsection{Proof of Theorem~\ref{thm:main}}\label{subsec:mainproof}

We first show that the assumptions of the theorem imply those of Proposition~\ref{prop:firstorder}.

By assumption, $\bphi$, $\{a_j,\bs_j\}_{j=1}^N$, and $\{b_j,\bt_j\}_{j=1}^M$ in the statement of the theorem is a solution of the following initial value problem (IVP) for some choice of initial conditions.

\begin{IVP}\label{IVP1}
Find $\bphi$, $\{a_j,\bs_j\}_{j=1}^N$, and $\{b_j,\bt_j\}_{j=1}^M$ such that
\begin{itemize}
\item \eqref{eq:sCM1}, \eqref{eq:sCM2}, and \eqref{eq:phidot} hold on a subset of $[0,T)$ containing $t=0$
\item the initial conditions
\begin{equation}\label{eq:picarddata1}
\begin{split}
&a_{j}(0)=a_{j,0},\qquad  \bs_j(0)=\bs_{j,0} \quad (j=1,\ldots,N), \\
&b_j(0)=b_{j,0},\qquad \bt_j(0)=\bt_{j,0} \quad  (j=1,\ldots,M),
\end{split}
\end{equation}
and 
\begin{equation}\label{eq:picarddata2}
 \dot{a}_j(0)=\dot{a}_{j,0} \quad (j=1,\ldots,N), \qquad \dot{b}_j(0)=\dot{b}_{j,0} \quad (j=1,\ldots,M)
\end{equation}
satisfy \eqref{eq:ajdot} at $t=0$
\end{itemize}
\end{IVP}

Because \eqref{eq:imaj}--\eqref{eq:ajakbjbk} hold, the functions of the form $F\big(\{a_j,\bs_j\}_{j=1}^N, \{b_j,\bt_j\}_{j=1}^M \big)$ defining the ODEs \eqref{eq:sCM1}, \eqref{eq:sCM2}, and \eqref{eq:phidot} are locally Lipshitz near the solution.  Standard uniqueness results on ODEs, for instance \cite[Theorem~4.18]{logemann2014}, thus guarantee that the solution is the unique solution of IVP~\ref{IVP1} on any interval $[0,T')\subseteq [0,T)$.

We will relate our known solution of IVP~\ref{IVP1} to a solution of the \eqref{eq:sCMs}, \eqref{eq:sCMt}, \eqref{eq:phidot}, and \eqref{eq:ajdot} in Theorem~\ref{thm:backlund}.

Note that when $\bs_j\neq \boldsymbol{0}$ and $\bt_j\neq \boldsymbol{0}$, \eqref{eq:ajdot} may be written as
\begin{equation}
\begin{split}\label{eq:ajdotIVP} 
\dot{a}_j=&\; -\frac{\bs_j}{\bs_j\cdot\bs_j^*}\wedge \Bigg(\ii\bphi-\sum_{k\neq j}^N\bs_k\zeta_2(a_j-a_k)+\sum_{k=1}^M \bt_k \zeta_2(a_j-b_k+\ii\delta)\Bigg)  \quad (j=1,\ldots,N),     \\
\dot{b}_j=&\; \frac{\bt_j}{\bt_j\cdot\bt_j^*}\wedge \Bigg(\ii\bphi+\sum_{k\neq j}^M \bt_k\zeta_2(b_j-b_k)-\sum_{k=1}^N \bs_k \zeta_2(b_j-a_k+\ii\delta)\Bigg) \quad (j=1,\ldots,M).
\end{split}
\end{equation}
We consider the following IVP. 

\begin{IVP}\label{IVP2}
Find $\bphi$, $\{a_j,\bs_j\}_{j=1}^N$, and $\{b_j,\bt_j\}_{j=1}^M$ such that
\begin{itemize}
\item \eqref{eq:sCMs}, \eqref{eq:sCMt}, \eqref{eq:phidot}, and \eqref{eq:ajdotIVP} hold on a subset of $[0,T)$ containing $t=0$
\item the initial conditions \eqref{eq:picarddata1} hold
\end{itemize}
\end{IVP}

By standard arguments (see, for instance, \cite[Chapters~4.2--4.3]{logemann2014}), IVP~\ref{IVP2} has a unique local solution which may be extended to a unique solution on a maximal interval $[0,T')\subseteq [0,T)$ where (i) the functions of the form $F\big(\{a_j,\bs_j\}_{j=1}^N, \{b_j,\bt_j\}_{j=1}^M \big)$ defining the ODEs  \eqref{eq:sCMs}, \eqref{eq:sCMt}, \eqref{eq:phidot}, and \eqref{eq:ajdotIVP} are locally Lipschitz; this is guaranteed when the conditions
\begin{equation}\label{eq:ajbk}
a_j-b_k+\ii\delta \neq 0 \bmod \Lambda \quad (j=1,\ldots,N,k=1,\ldots,M),
\end{equation}
where 
\begin{equation}\label{eq:Lambda}
\Lambda\coloneqq \{2n\ell+2m\ii\delta: n,m\in \Z\},
\end{equation}
\eqref{eq:ajakbjbk} ($\mathrm{mod}\,\Lambda$), and \eqref{eq:sjtj} hold and (ii) each function $\bphi$, $\{a_j,\bs_j\}_{j=1}^N$, and $\{b_j,\bt_j\}_{j=1}^M$ remains finite. In the case where $T'<T$, one of the conditions (i), (ii) must be violated as $t\to T'$. 

We denote the maximal solution of IVP~\ref{IVP2} by $\hat{\bphi}$, $\{\hat{a}_j,\hat{\bs}_j\}_{j=1}^N$, and $\{\hat{b}_j,\hat{\bt}_j\}_{j=1}^M$. Theorem~\ref{thm:backlund} shows that this solution of IVP~\ref{IVP2} is also a solution of IVP~\ref{IVP1} on $[0,T')$. 
Suppose $T'<T$. We now consider two cases. 

In the first case, suppose that each of the quantities $|\hat{a}_j-\hat{b}_k+2n\ell +(2m+1)\ii\delta|$, $\hat{\bs}_j\cdot\hat{\bs}_j^*$, and $\mspace{-1.5mu}\hat{\mspace{1.5mu} \bt}_k\cdot\mspace{-1.5mu}\hat{\mspace{1.5mu} \bt}_k^*$ (where $j=1,\ldots,N$, $k=1,\ldots,M$, $n,m\in\Z$, and $|\cdot|$ is the modulus) is bounded from below by some $\epsilon>0$ on $[0,T')$. By Theorem~\ref{thm:backlund}, this gives a solution of IVP~\ref{IVP1} that either violates \eqref{eq:ajakbjbk} or becomes unbounded as $t\to T'$. We have constructed a maximal solution of IVP~\ref{IVP1} on a proper subinterval of $[0,T)$, a contradiction.

In the second case, suppose that either \eqref{eq:ajbk} is violated or least one of the quantities $\hat{\bs}_j\cdot\hat{\bs}_j^*$ and $\mspace{-1.5mu}\hat{\mspace{1.5mu} \bt}_k\cdot\mspace{-1.5mu}\hat{\mspace{1.5mu} \bt}_k^*$ tends to $0$ as $t\to T'$. It follows that either \eqref{eq:imaj} or \eqref{eq:sjtj} fails to hold in the limit $t\to T'$. By Theorem~\ref{thm:backlund}, we have a solution of IVP~\ref{IVP1} on $[0,T')$ such that  \eqref{eq:imaj} or \eqref{eq:sjtj} is violated as $t\to T'$.
Because the known solution of IVP~\ref{IVP1} is unique and satisfies \eqref{eq:imaj} and \eqref{eq:sjtj} on $[0,T)$, this is a contradiction. 

We conclude that $T'=T$ and so IVP~\ref{IVP2} admits a unique maximal solution on $[0,T)$. By Theorem~\ref{thm:backlund} and the uniqueness of the known solution $\bphi$, $\{a_j,\bs_j\}_{j=1}^N$, and $\{b_j,\bt_j\}_{j=1}^M$ to IVP~\ref{IVP1}, we see that this known solution solves IVP~\ref{IVP2} on $[0,T)$. It follows that assumptions of Proposition~\ref{prop:firstorder} are satisfied. 

We now use Propositions~\ref{prop:constraints}, \ref{prop:firstorder}, and \ref{prop:conserved} to show that under the conditions of the theorem, the ansatz \eqref{eq:ansatz} solves the periodic ncIHF equation \eqref{eq:ncIHF} and each component of this solution has constant length $\rho$.

By Proposition~\ref{prop:firstorder}, \eqref{eq:ansatz} with the solution $\bphi$, $\{a_j,\bs_j\}_{j=1}^N$, and $\{b_j,\bt_j\}_{j=1}^M$ of IVP~\ref{IVP1} solves the periodic ncIHF equation at all $t\in [0,T)$ where the derivatives of $\bu$ and $\bv$ with respect to $x$ and $t$ exist. It remains to show that this solution satisfies $\bu(x,t)^2=\bv(x,t)^2=\rho^2$ for all $x\in [-\ell,\ell)$ and $t\in [0,T)$. By Proposition~\ref{prop:constraints}, this holds provided \eqref{eq:constraint1}--\eqref{eq:constraint3} and \eqref{eq:constraint4} are satisfied. Each of the constraints \eqref{eq:constraint1}--\eqref{eq:constraint3} and \eqref{eq:constraint4} are satisfied at $t=0$ by assumption. By Lemma~\ref{lem:totalspin} and Proposition~\ref{prop:conserved}, these constraints hold on $[0,T)$ and hence the theorem follows. 

\section{Examples of solutions}\label{sec:explicit}

We construct examples of solutions of the periodic ncIHF equation \eqref{eq:ncIHF} using Theorem~\ref{thm:main}. Analogs of one-soliton traveling wave solutions known for the ncIHF equation on the real line \cite{berntsonklabbers2020} are given in Section~\ref{subsec:tw}. We use an elliptic parameterization of $S^2$ to construct a class of real initial data for the periodic ncIHF equation satisfying the constraints of Theorem~\ref{thm:main} in Section~\ref{subsec:parameterization}. The results of Section~\ref{subsec:parameterization} are used to obtain a breather-type solution of the periodic ncIHF equation in Section~\ref{subsec:breather}. 

Sections~\ref{subsec:parameterization}--\ref{subsec:breather} are concerned with real-valued solutions
\begin{align}\label{eq:ansatzreal}
\left(\begin{array}{c} \bu(x,t) \\ \bv(x,t)   \end{array}\right)=  \bphi(t)\left(\begin{array}{c} 1 \\ 1    \end{array}\right) &\; +\ii \sum_{j=1}^N \bs_j(t)\left(\begin{array}{c} \zeta_2(x-a_j(t)+\ii\delta/2) \\ \zeta_2(x-a_j(t)-\ii\delta/2)    \end{array}\right) \nonumber \\
&\; -\ii \sum_{j=1}^N \bs_j^*(t) \left(\begin{array}{c} \zeta_2(x-a_j^*(t)-\ii\delta/2) \\ \zeta_2(x-a_j^*(t)+\ii\delta/2)    \end{array}\right), 
\end{align}
of the periodic ncIHF equation satisfying $\bu(x,t)=\bv(x,t)^2=1$. Such solutions are characterized by the following consistent reduction of Theorem~\ref{thm:main} where
\begin{equation}\label{eq:reduction}
N=M, \qquad \rho=1,\qquad  \bphi^*=\bphi, \qquad b_j=a_j^*, \qquad \bt_j=\bs_j^* \quad (j=1,\ldots,N). 
\end{equation}

\begin{corollary}\label{cor:main}
For $N\in \Z_{\geq 1}$ and $T>0$, let $\bphi$ and $\{a_j,\bs_j\}_{j=1}^N$ be a solution of the equations \eqref{eq:sCM1} and
\begin{equation}\label{eq:phidotreal}
\dbphi=\frac{\ii}{2} \sum_{j=1}^N\sum_{k\neq j}^N \bs_j\wedge\bs_k f_2'(a_j-a_k)- \frac{\ii}{2}\sum_{j=1}^N\sum_{k\neq j}^N \bs_j^*\wedge\bs_k^* f_2'(a_j^*-a_k^*)
\end{equation}
on the interval $[0,T)$ with initial conditions that satisfy the following equations at $t=0$,
\begin{align}
&\bs_j^2=0, \\
&\bs_j\cdot\Bigg(\ii\bphi-\sum_{k\neq j}^{N} \bs_k\zeta_2(a_j-a_k)+\sum_{k=1}^N \bs_k^* \zeta_2(a_j-a_k^*+\ii\delta)\Bigg)=0, \\
&\bs_j \dot{a}_j = -\bs_j\wedge\Bigg(\ii\bphi-\sum_{k\neq j}^N \bs_k\zeta_2(a_j-a_k)+\sum_{k=1}^N \zeta_2(a_j-a_k^*+\ii\delta)\Bigg) \label{eq:ajdotreal},
\end{align}
for $j=1,\ldots,N$ and 
\begin{equation}
\bphi^2=1+\frac12 \sum_{j=1}^N\sum_{k\neq j}^N \bs_j\cdot\bs_k f_2(a_j-a_k)+\frac12\sum_{j=1}^N\sum_{k\neq j}^N \bs_j^*\cdot\bs_k^* f_2(a_j^*-a_k^*)-\sum_{j=1}^N\sum_{k=1}^N \bs_j\cdot\bs_k^*f_2(a_j-a_k^*+\ii\delta).
\end{equation}
Moreover, suppose that the conditions
\begin{equation}
\frac{\delta}{2}< \im\,a_j(t)<\frac{3\delta}{2} \quad (j=1,\ldots,N),\qquad a_j(t)\neq a_k(t)\quad (1\leq j<k\leq N)
\end{equation}
hold for $t\in [0,T)$. Then, for all $t\in [0,T)$ such that the functions $\bu(x,t)$ and $\bv(x,t)$ in \eqref{eq:ansatzreal} are differentiable with respect to $x$ and $t$ for all $x\in [-\ell,\ell)$, \eqref{eq:ansatzreal} provides an exact solution to the periodic ncIHF equation \eqref{eq:ncIHF} satisfying $\bu(x,t)^2=\bv(x,t)^2=1$. 
\end{corollary}

\subsection{Solutions that are sums of traveling waves}\label{subsec:tw}

As mentioned in Remark~\ref{rem:complex}, Theorem~\ref{thm:main} does not include real traveling wave solutions when $N=M=1$. In fact, as we will show, the class of solutions with $N=M=1$ does not contain any traveling waves but instead consists of solutions that are the sums of two traveling waves moving in opposite directions.

When $N=M=1$, the constraint \eqref{eq:constraint3} implies that $\bs_1=\bt_1$. Consequently, the constraints \eqref{eq:constraint1}, \eqref{eq:constraint2}, and \eqref{eq:constraint4} at $t=0$ are reduced to
\begin{equation}\label{eq:twconstraints}
\bs_{1,0}^2=0,\qquad \bs_{1,0}\cdot \bphi_0=0, \qquad \bphi_0^2=\rho^2,
\end{equation}
respectively. The general solution of the first constraint in \eqref{eq:twconstraints} is \cite[Lemma~B.1]{berntsonklabbers2020}
\begin{equation}\label{eq:sgensol}
\bs_{1,0}=s_{1,0}(\bn_1+\ii\bn_2),
\end{equation}
where $s_{1,0}\in \C$ and $\bn_1,\bn_2\in S^2$ such that $\bn_1\cdot\bn_2=0$. By expanding $\bphi_0$ in the basis $\{\bn_1,\bn_2,\bn_1\wedge\bn_2\}$ of $\C^3$, $\bphi_0=\phi_{0,1}\bn_1+\phi_{0,2}\bn_2+\phi_{0,12}\bn_{1}\wedge\bn_{2}$, we see that the general solution of the second and third constraints in \eqref{eq:twconstraints} with \eqref{eq:sgensol} is
\begin{equation}\label{eq:phi0}
\bphi_0=\phi_{1,0}(\bn_1+\ii\bn_2)+\rho \bn_1\wedge\bn_2,
\end{equation}
with $\phi_{1,0}\in\C$ arbitrary. Moreover, the equations of motion \eqref{eq:phidot}, \eqref{eq:sCM1}, and \eqref{eq:sCM2} reduce to 
\begin{equation}\label{eq:tweom}
\dbphi=\boldsymbol{0}, \qquad \dot{\bs}_1=0, \qquad \ddot{a}_1=0, \qquad \ddot{b}_1=0.
\end{equation}
which may be integrated to $\bphi=\bphi_0$, $\bs_1=\bs_{1,0}$, $a_1=a_{1,0}+\dot{a}_{1,0} t$, $b_1=b_{1,0}+\dot{b}_{1,0} t$. Imposing \eqref{eq:ajdot} and using $\bs_{1,0}\wedge(\ii\bphi_0)=-\rho\bs_{1,0}$, which follows from \eqref{eq:VecId} with \eqref{eq:sgensol}--\eqref{eq:phi0}, gives $\dot{a}_{1,0}=\rho$ and $\dot{b}_{1,0}=-\rho$. We have arrived at the following class of exact solutions of the periodic ncIHF equation \eqref{eq:ncIHF},
\begin{align}\label{eq:tw1}
\left(\begin{array}{c} \bu(x,t) \\ \bv(x,t) \end{array}\right)=&\;  \big(\phi_{1,0}(\bn_1+\ii\bn_2)+\rho \bn_1\wedge\bn_2\big)\left(\begin{array}{c} 1 \\ 1 \end{array}\right) \nonumber \\
&\; +\ii\big(s_{1,0}(\bn_1+\ii\bn_2)\big)\left(\begin{array}{c} \zeta_2(x-a_{1,0}-\rho t+\ii\delta/2)-\zeta_2(x-b_{1,0}+\rho t-\ii\delta/2) \\   \zeta_2(x-a_{1,0}-\rho t-\ii\delta/2)-\zeta_2(x-b_{1,0}+\rho t+\ii\delta/2)   \end{array}\right),
\end{align}
where $\phi_{1,0},s_{1,0},\rho\in \C$ and $\bn_1,\bn_2\in S^2$ are arbitrary and $a_{1,0}$ and $b_{1,0}$ satisfy \eqref{eq:ajakbjbk} (at $t=0$). We note that in the case $\rho\notin \R$, the condition \eqref{eq:imaj} will be violated in finite time, after which Theorem~\ref{thm:main} does not guarantee that \eqref{eq:tw1} provides a solution.

The argument above can be generalized to give the following solutions in the case $N=M$ ($N\geq 1$),
\begin{align}\label{eq:tw2}
\left(\begin{array}{c} \bu(x,t) \\ \bv(x,t) \end{array}\right)=&\;  \big(\phi_{1,0}(\bn_1+\ii\bn_2)+\rho \bn_1\wedge\bn_2\big)\left(\begin{array}{c} 1 \\ 1 \end{array}\right) \nonumber \\
&\; +\ii\big(s_{1,0}(\bn_1+\ii\bn_2)\big)\sum_{j=1}^N \left(\begin{array}{c} \zeta_2(x-a_{j,0}-\rho t+\ii\delta/2)-\zeta_2(x-b_{j,0}+\rho t-\ii\delta/2) \\   \zeta_2(x-a_{j,0}-\rho t-\ii\delta/2)-\zeta_2(x-b_{j,0}+\rho t+\ii\delta/2)   \end{array}\right).
\end{align}
where $\phi_{1,0},s_{1,0},\rho\in \C$ and $\bn_1,\bn_2\in S^2$ are arbitrary and the $a_{j,0}$ and $b_{j,0}$ ($j=1,\ldots,N$) satisfy \eqref{eq:ajakbjbk} (at $t=0$). Similar remarks as above, concerning the finite-time existence of the solution \eqref{eq:tw2} with $\rho\notin \R$, apply.

\begin{remark}
The absence of nontrivial traveling wave solutions for the periodic ncIHF equation in Theorem~\ref{thm:main} is surprising in view of the the rich structure of analogous solutions, obtainable via pole ansatz \cite{berntsonklabbers2020}, for the half wave maps equation \cite{zhou2015,lenzmann2018}. We regard the classification of traveling wave solutions of the periodic ncIHF equation as an interesting open problem. 
\end{remark}

\subsection{Initial data from an elliptic parameterization of the two-sphere}\label{subsec:parameterization}

One way to find initial data satisfying \eqref{eq:constraint1}--\eqref{eq:constraint3} and \eqref{eq:constraint4} is by considering the following parameterization of the two-sphere, defined by a map from $\R^2$ to $S^2$, 
\newcommand{\Q}{\mathbb{Q}}
\newcommand{\N}{\mathbb{N}}
\begin{equation}\label{eq:R2toS2}
(x_1,x_2)\mapsto\big( \sn(x_1 | m) \cn(x_2 |m), \sn(x_1|m)\sn(x_2|m), \cn(x_1|m) \big),
\end{equation} 
where $\sn(\cdot|m)$ and $\cn(\cdot|m)$ are the Jacobi sine and cosine functions with elliptic modulus $m$. The $S^2$-valuedness of \eqref{eq:R2toS2} can be shown using the identity \eqref{eq:sn2cn2}. Requisite details on the functions $\sn(z|m)$ and $\cn(z|m)$ and the elliptic integrals\footnote{In this context only, the prime in $K'=K'(m)$ does not indicate differentiation with respect to the argument; see \eqref{eq:Kp} for the definition of this function.} $K=K(m)$ and $K'=K'(m)$, which determine the periods of Jacobi elliptic functions, can be found in Appendix~\ref{app:elliptic}. The functions $\sn(z|m)$ and $\cn(z|m)$ are elliptic functions of $z$ with half-periods $(2K,\ii K')$ and $(2K,K+\ii K')$, respectively. Both functions have simple poles at 
\begin{equation}\label{eq:xijk}
\xi_{jk} \coloneqq 2j K + (2k+1)\ii K' \quad  (j,k\in \Z)
\end{equation}
with corresponding residues 
\begin{equation}\label{eq:residues}
\underset{z=\xi_{jk}}{\mathrm{Res}}\sn(z|m)=\frac{(-1)^j}{\sqrt{m}},\qquad \underset{z=\xi_{jk}}{\mathrm{Res}}\cn(z|m)=-\frac{\ii(-1)^{j+k}}{\sqrt{m}}.
\end{equation}

Below, in Proposition~\ref{prop:jacobi}, we show that a specialization of the map \eqref{eq:R2toS2},
\begin{equation}\label{eq:r}
\br(x)\coloneqq \big(\sn(px|m)\cn(q(x-x_0)|m),\sn(px|m)\sn(q(x-x_0)|m),\cn(px|m)\big),
\end{equation}
for positive integers $p,q$ and real $x_0$, can be used to construct real initial data satisfying the conditions of Theorem~\ref{thm:main}, where the parameters $\bphi_0$, $\bs_{j,0}$, and $a_{j,0}$ satisfy the constraints \eqref{eq:constraint1}--\eqref{eq:constraint3} and \eqref{eq:constraint4} with $N=M$, $\bt_j=\bs_j^*$, and $b_j=a_j^*$ at $t=0$. Note that the primitive periods of the function $\br(x)$ in \eqref{eq:r} are $4K(m)$ and $4\ii K'(m)$ and thus we set $\ell=2K(m)$, $\ii\delta=2\ii K'(m)$ as the half-periods of the $\zeta_2$-function in \eqref{eq:ansatzreal}.
\begin{proposition}\label{prop:jacobi}
Let $m \in (0,1)$, $p,q \in \Z_{\geq 1}$, and $x_0\in (0,4K(m))$ such that the sets
\begin{equation}\label{eq:A1}
\mathcal{A}_1\coloneqq \big\{\alpha_{jk}^{(1)}:0\leq j\leq 2p-1,0\leq k\leq p-1\big\},\qquad \alpha_{jk}^{(1)}\coloneqq \frac{\xi_{jk}}{p}
\end{equation}
and 
\begin{equation}\label{eq:A2}
\mathcal{A}_2\coloneqq \big\{\alpha_{jk}^{(2)}:0\leq j\leq 2q-1,0\leq k\leq q-1\big\},\qquad \alpha_{jk}^{(2)}\coloneqq \frac{\xi_{jk}}{q}+x_0
\end{equation}
are disjoint. Then, \eqref{eq:ansatzreal} with $N=2(p^2+q^2)$, $\ell=2K(m)$, $\ii\delta=2\ii K'(m)$,
\begin{equation}\label{eq:ajsj1}
\begin{split}
a_{jp+k+1,0}=&\; \alpha_{jk}^{(1)}+\frac{\ii\delta}{2}, \\
\bs_{jp+k+1,0}=&\; \frac{-\ii (-1)^j}{ p \sqrt{m}}\left(\begin{array}{c}
\cn\big(q\big(\alpha_{jk}^{(1)}-x_0\big)\big) \\
\sn\big(q\big(\alpha_{jk}^{(1)}-x_0\big)\big) \\
-\ii (-1)^{k}
\end{array}\right)
\end{split} \quad (0\leq j\leq 2p-1,0\leq k\leq p-1),
\end{equation}
\begin{equation}\label{eq:ajsj2}
\begin{split}
a_{2p^2+j q+k+1,0}=&\; \alpha_{jk}^{(2)}+\frac{\ii\delta}{2}, \\
\bs_{2p^2+j q+k+1,0}=&\; \frac{-\ii (-1)^{j}}{ q \sqrt{m}}\left(\begin{array}{c}
-\ii(-1)^k\sn\big( p \alpha_{jk}^{(2)}  \big) \\
\sn\big( p \alpha_{jk}^{(2)}\big)  \\
0
\end{array}\right)
\end{split} \quad (0\leq j\leq 2q-1,0\leq k\leq q-1),
\end{equation}
and
\begin{align}\label{eq:phi02}
\bphi_0= \left(\begin{array}{c} 0 \\ 0 \\ 1 \end{array}\right)+\frac{1}{\sqrt{m}} &\left(\sum_{j=0}^{2p-1}\sum_{k=0}^{p-1} \frac{(-1)^j}{p}\zeta_2\big(\alpha_{jk}^{(1)}\big)\left(\begin{array}{c}
\cn\big(q\big(\alpha_{jk}^{(1)}-x_0\big)\big) \\
\sn\big(q\big(\alpha_{jk}^{(1)}-x_0\big)\big) \\
-\ii(-1)^k
 \end{array}\right)\right. \nonumber \\
 & \left.+\sum_{j=0}^{2q-1}\sum_{k=0}^{q-1} \frac{(-1)^j}{q}\zeta_2\big(\alpha_{jk}^{(2)}\big)\left(\begin{array}{c}
-\ii(-1)^k\sn\big(p\alpha_{jk}^{(2)}\big) \\
\sn\big(p\alpha_{jk}^{(2)}\big) \\
0
 \end{array}\right)+\mathrm{c.c.}\right)
\end{align}
(where $\mathrm{c.c.}$ denotes the complex conjugate of the terms within the parentheses) provides initial data for the periodic ncIHF equation satisfying the conditions of Theorem~\ref{thm:main}.
\end{proposition}

\begin{remark}
The sets $\mathcal{A}_1$ and $\mathcal{A}_2$ contain the poles of the functions $\sn(px|m)$ and $\cn(q(x-x_0)|m)$ (equivalently $\sn(q(x-x_0)|m)$, see \eqref{eq:xijk}), respectively for $x$ in $[0,2\ell)\times \ii[0,\delta)$. The corresponding full sets of poles within the period parallelogram $[0,2\ell)\times \ii[-\delta,\delta)$ are $\mathcal{A}_1\cup \mathcal{A}_1^*$ and $\mathcal{A}_2\cup \mathcal{A}_2^*$, respectively. In view of \eqref{eq:xijk}, it is natural to label the elements of $\mathcal{A}_1$ by non-negative integers $j,k$ satisfying $0\leq j\leq 2p-1$, $0\leq k\leq p-1$ and similarly for $\mathcal{A}_2$, see \eqref{eq:A1}-\eqref{eq:A2}. On the other hand, the poles $a_j$ in \eqref{eq:ansatzreal} are labelled by a single index $j\in \{1,\ldots,N=2(p^2+q^2)\}$. To bridge this gap, the subscripts in \eqref{eq:ajsj1}--\eqref{eq:ajsj2} define a bijection between the underlying index sets of (i) the $\alpha_{jk}^{(1)}$, $\alpha_{jk}^{(2)}$ in \eqref{eq:A1}--\eqref{eq:A2} and (ii) the $a_j$ in \eqref{eq:ansatzreal}.
\end{remark}

\subsubsection{Proof of Proposition~\ref{prop:jacobi}}

We begin by writing the function $\br(x)$ in \eqref{eq:r} in terms of the function $\zeta_2(z)$ \eqref{eq:zeta2}.

\begin{lemma}\label{lem:jacobi}
The components $\big(r^1(x),r^2(x),r^3(x)\big)$ of the function $\br(x)$ in \eqref{eq:r} can be decomposed in terms of the function $\zeta_2(z;\ell,\ii\delta)$ with half-periods $\ell=2K(m)$ and $\ii\delta = 2\ii K'(m)$ as follows,
\begin{align}
\label{eq:jacobi_decomposition}
r^1(x)=&\;  \frac{1}{\sqrt{m}}\sum_{j=0}^{2p-1} \sum_{k=0}^{p-1}  \frac{(-1)^{j}}{p} \cn \big(q\big(\alpha_{jk}^{(1)} -x_0\big)\big)\big( \zeta_2\big(x-\alpha_{jk}^{(1)}    \big) +\zeta_2\big(\alpha_{jk}^{(1)}\big) \big)  \nonumber \\
&\;- \frac{\ii}{\sqrt{m}} \sum_{j=0}^{2q-1} \sum_{k=0}^{q-1} \frac{(-1)^{j+k}}{q} \sn\big(p \alpha_{jk}^{(2)}\big) \big( \zeta_2\big(x-\alpha_{jk}^{(2)}   \big)+\zeta_2\big( \alpha_{jk}^{(2)}    \big)\big) + \mathrm{c.c.},  \nonumber  \\
r^2(x) =&\;  \frac{1}{\sqrt{m}}\sum_{j=0}^{2p-1} \sum_{k=0}^{p-1} \frac{(-1)^{j}}{p} \sn\big(q\big(\alpha_{jk}^{(1)} -x_0\big)\big)\big( \zeta_2\big(x-\alpha_{jk}^{(1)} \big) +\zeta_2\big(\alpha_{jk}^{(1)}\big) \big)   \nonumber \\
&\;+ \frac{1}{\sqrt{m}}\sum_{j=0}^{2q-1} \sum_{k=0}^{q-1} \frac{(-1)^{j}}{q} \sn\big(p \big( \alpha_{jk}^{(2)} \big) \big) \big( \zeta_2\big(x-\alpha_{jk}^{(2)}\big) +\zeta_2\big(\alpha_{jk}^{(2)}\big) \big) + \mathrm{c.c.},   \nonumber \\
r^3(x)-1=&\; -\frac{\ii}{\sqrt{m}} \sum_{j=0}^{2p-1} \sum_{k=0}^{p-1}\frac{(-1)^{j+k}}{p} \big(\zeta_2\big(x-\alpha_{jk}^{(1)}\big)+\zeta_2\big(\alpha_{jk}^{(1)}\big)\big)+ \mathrm{c.c.},
\end{align}
where $\mathrm{c.c.}$ denotes the complex conjugate of the written terms.
\end{lemma}

\begin{proof}
We consider the function $\br(z)$ for $z\in \Pi\coloneqq [0,2\ell)\times \ii[-\delta,\delta)$, a (primitive) period parallelogram. The function $\br(z)$ has a pole at each element of $\mathcal{A}_1\cup \mathcal{A}_2\cup \mathcal{A}_1^*\cup \mathcal{A}_2^*$, with $\mathcal{A}_1$ and $\mathcal{A}_2$ defined in \eqref{eq:A1}--\eqref{eq:A2}. It follows from \eqref{eq:xijk}, \eqref{eq:residues}, and the definition of $\br(x)$ \eqref{eq:r} that
\begin{equation}\label{eq:residuesymmetry}
\underset{z=\big(\alpha_{jk}^{(1)}\big)^*}{\mathrm{Res}} \br(z)=\bigg(\underset{z=\alpha_{jk}^{(1)}}{\mathrm{Res}} \br(z)\bigg)^*,\qquad \underset{z=\big(\alpha_{jk}^{(2)}\big)^*}{\mathrm{Res}} \br(z)=\bigg(\underset{z=\alpha_{jk}^{(2)}}{\mathrm{Res}} \br(z)\bigg)^*.
\end{equation}
We will use this symmetry to obtain \eqref{eq:jacobi_decomposition}. Let
\begin{equation}\label{eq:g2}
g(z;\alpha)\coloneqq \zeta_2(z-\alpha)-\zeta_2(\alpha),
\end{equation}
a $2\ii\delta$-periodic meromorphic function of $z$ with simple poles at $z=\alpha \bmod \Lambda$, with $\Lambda$ as in \eqref{eq:Lambda}. Additionally, $g(0;\alpha)=0$ when $\alpha\neq 0\bmod\Lambda$; this follows from \eqref{eq:g2} and the fact that $\zeta(z)$ is an odd function \eqref{eq:parity}.

We claim that
\begin{equation}\label{eq:ransatz}
\br(x)-\br(0)=\sum_{j=0}^{2p-1}\sum_{k=0}^{p-1} \bigg(\underset{z=\alpha_{jk}^{(1)}}{\mathrm{Res}} \br(z)\bigg)  g\big(x;\alpha_{jk}^{(1)}\big)+\sum_{j=0}^{2q-1}\sum_{k=0}^{q-1} \bigg(\underset{z=\alpha_{jk}^{(2)}}{\mathrm{Res}} \br(z)\bigg)  g\big(x;\alpha_{jk}^{(2)}\big)+\mathrm{c.c.}
\end{equation}
To see this, we note that, by \eqref{eq:residuesymmetry}, the right-hand-side of \eqref{eq:ransatz} has the same poles and residues as $\br(z)-\br(0)$ within $\Pi$. Moreover, the right-hand-side of \eqref{eq:ransatz} is elliptic: while $2\ii\delta$-periodicity follows from that of $g(z;\alpha)$ via \eqref{eq:imperiod}, $2\ell$-periodicity is a consequence of the identities $g(z+2\ell;\alpha)=g(z;\alpha)+\pi/\delta$, which follows from \eqref{eq:realperiod}, and 
\begin{equation}
\sum_{j=0}^{2p-1}\sum_{k=0}^{p-1} \underset{z=\alpha_{jk}^{(1)}}{\mathrm{Res}} \br(z) +\sum_{j=0}^{2q-1}\sum_{k=0}^{q-1} \underset{z=\alpha_{jk}^{(2)}}{\mathrm{Res}} \br(z) +\mathrm{c.c.}=\boldsymbol{0},
\end{equation}
which holds using \eqref{eq:residuesymmetry} and the fact that the sum of residues within $\Pi$ of the elliptic function $\br(z)$ vanishes. Because both sides of \eqref{eq:ransatz} evaluate to $\boldsymbol{0}$ at $x=0$, by Liouville's theorem, \eqref{eq:ransatz} holds.

The result \eqref{eq:jacobi_decomposition} follows from \eqref{eq:ransatz} after inserting $\br(0)=(0,0,1)$ (because $\sn(0|m)=0$ and $\cn(0|m)=1$) and computing the residues using \eqref{eq:r} and \eqref{eq:residues}.
\end{proof}

We set $\bu_0(x)=\br(x)$. By comparing \eqref{eq:ansatzreal} with \eqref{eq:jacobi_decomposition}, we obtain \eqref{eq:ajsj1}--\eqref{eq:ajsj2}, and \eqref{eq:phi02}. Because $\br(x)^2=1$ by construction, we have $\bu_0(x)^2=1$. From \eqref{eq:UdotU3}, it is clear that $\bv_0(x)^2=1$ (with $\bv_0(x)$ given by \eqref{eq:ansatzreal} with \eqref{eq:ajsj1}--\eqref{eq:ajsj2}, and \eqref{eq:phi02}) if and only if $\bu_0(x)^2=1$. We now apply Proposition~\ref{prop:constraints} directly in the special case $N=M$, $\bt_j=\bs_j^*$, $b_j=a_j^*$, and $\rho=1$. Because $\bu_0(x)^2=\bv_0(x)^2=1$, we have that the constraints \eqref{eq:constraint1}-\eqref{eq:constraint3} and \eqref{eq:constraint4} are satisfied by \eqref{eq:ajsj1}--\eqref{eq:ajsj2}, and \eqref{eq:phi02}.

\subsubsection{Numerical implementation}\label{subsubsec:numerics}

In the source file of our submission, we have included a Mathematica notebook to visualize solutions of the periodic ncIHF equation with initial data in the form \eqref{eq:r}. Using Proposition~\ref{prop:jacobi}, such data may be transformed into the form \eqref{eq:ansatzreal} (a special case of \eqref{eq:ansatz}) to which Theorem~\ref{thm:main} applies. For chosen $p$, $q$, $m$, and $x_0$, our Mathematica notebook performs the transformation of Proposition~\ref{prop:jacobi} and uses the resulting parameters $a_{j,0}$, $\bs_{j,0}$, and $\bphi_0$ as initial conditions for the reduction \eqref{eq:reduction} of the ODE system in Theorem~\ref{thm:main}. By numerically solving these ODEs, we obtain numerical solutions of the periodic ncIHF equation in the form \eqref{eq:ansatzreal}. Visualizations of a particular solution obtained using this method are presented in Section~\ref{subsec:breather}.

\subsection{A breather solution}\label{subsec:breather}

We study a particular instance of the solution of the periodic ncIHF equation with initial data constructed using Proposition~\ref{prop:jacobi}. This solution exhibits energy oscillations reminiscent of well-known breather solutions of the nonlinear Schr\"{o}dinger \cite{akhmediev1986} and sine-Gordon equations \cite{ablowitz1973}. To be more specific, we will present numerical evidence of a solution of the ncIHF equation where the energy density is time-periodic but the solution itself is not. An explicit formula for the energy density of a solution \eqref{eq:ansatzreal} of the ncIHF equation is presented in Section~\ref{sec:energydensity}.

To avoid misunderstanding, we emphasize that the results presented in this subsection are primarily numerical:  a particular exact solution of the constraints \eqref{eq:constraint1}--\eqref{eq:constraint3} and \eqref{eq:constraint4} given in Section~\ref{subsec:parameterization} provides admissible initial data for Theorem~\ref{thm:main}; we numerically solve the equations of motion of the spin CM system \eqref{eq:sCM1} and background dynamics \eqref{eq:phidot} to evolve the solution \eqref{eq:ansatz} in time, using the method described in Section~\ref{subsubsec:numerics}. 

We set $p=q=1$ and $x_0=K$ in \eqref{eq:r} to obtain the following map from $\R$ to $S^2$. 
\begin{equation}\label{eq:rbreather}
\br(x)\coloneqq \big(\sn(x|m)\cn(x-K|m),\sn(x|m)\sn(x-K|m),\cn(x|m)\big).
\end{equation}
We set $m=1/2$, yielding $\ell=\delta=2K(1/2) \approx 3.708$. Using Proposition~\ref{prop:jacobi}, \eqref{eq:rbreather} can be written as \eqref{eq:ansatzreal} with $N=4$ and
\begin{align}\label{eq:breatherdata}
&a_{1,0}=2\ii K(1/2), \qquad a_{2,0}=(2+2\ii) K(1/2), \qquad a_{3,0}=(1+2\ii) K(1/2),\qquad a_{4,0}=(3+2\ii) K(1/2), \nonumber \\
&\bs_{1,0}=\big(\sqrt{2},2\ii,-\sqrt{2}\big),\qquad \bs_{2,0}=\big(\sqrt{2},2\ii,-\sqrt{2}\big),\qquad \bs_{3,0}=(-2,-2\ii,0),\qquad \bs_{4,0}=(-2,-2\ii,0), \nonumber \\
&\bphi_0\approx (0,1.694,0).
\end{align}

In accordance with Corollary~\ref{cor:main}, we solve \eqref{eq:sCM1} and \eqref{eq:phidotreal} subject to the initial conditions \eqref{eq:breatherdata} and with initial velocities computed from \eqref{eq:ajdotreal} (at $t=0$). The resulting dynamics for the poles $a_j$ are time-periodic with period $T\approx 11.83$. A visualization of the dynamics of the poles is shown in Fig.~\ref{fig:breather_poles}.

\begin{figure}
\centering
\begin{tikzpicture}[scale=0.9]
\def\a{2.25};
\def\b{2.5};
\def\d{4.5};
\node at (\d,0)
{
\begin{tikzpicture}
\node at (0,0) {\includegraphics[scale=0.65]{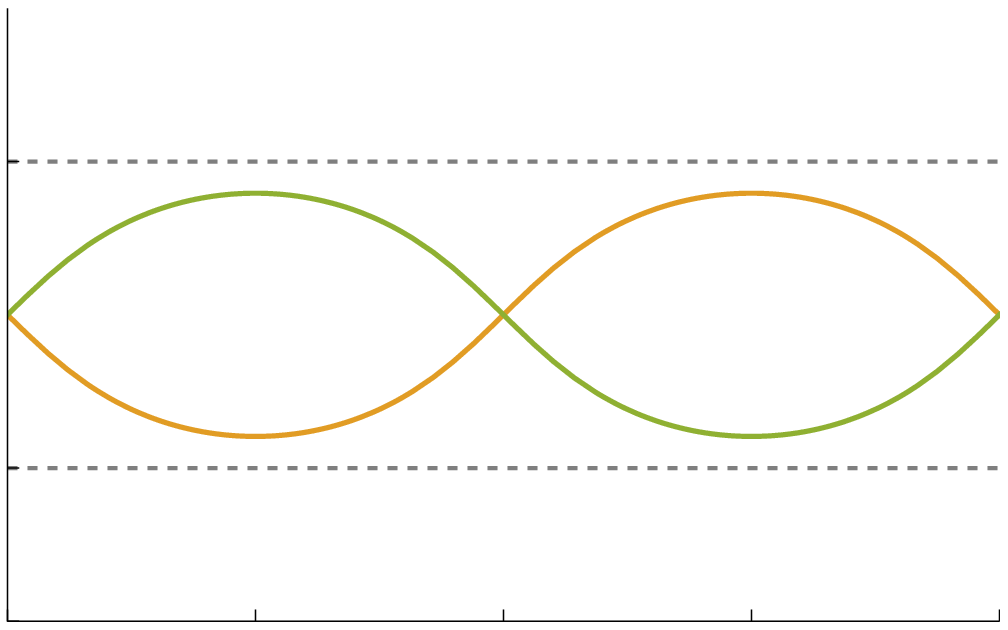}
};
\node at (-3.3,-\a+0.07) {\scriptsize{$0$}};
\node at (-1.6,-\a) {\scriptsize{$T/4$}};
\node at (0,-\a) {\scriptsize{$T/2$}};
\node at (1.6,-\a) {\scriptsize{$3T/4$}};
\node at (3.25,-\a) {\scriptsize{$T$}};

\node at (3.4,-1.65) {$t$};
\node at (-3.05,2.2) {\small $\mathrm{Im}\,a$};

\node at (-3.6,1) {\scriptsize $3\delta/2$};
\node at (-3.55,-1) {\scriptsize $\delta/2$};

\end{tikzpicture}
};
\node at (-\d,0)
{
\begin{tikzpicture}
\node at (0,0) {\includegraphics[scale=0.65]{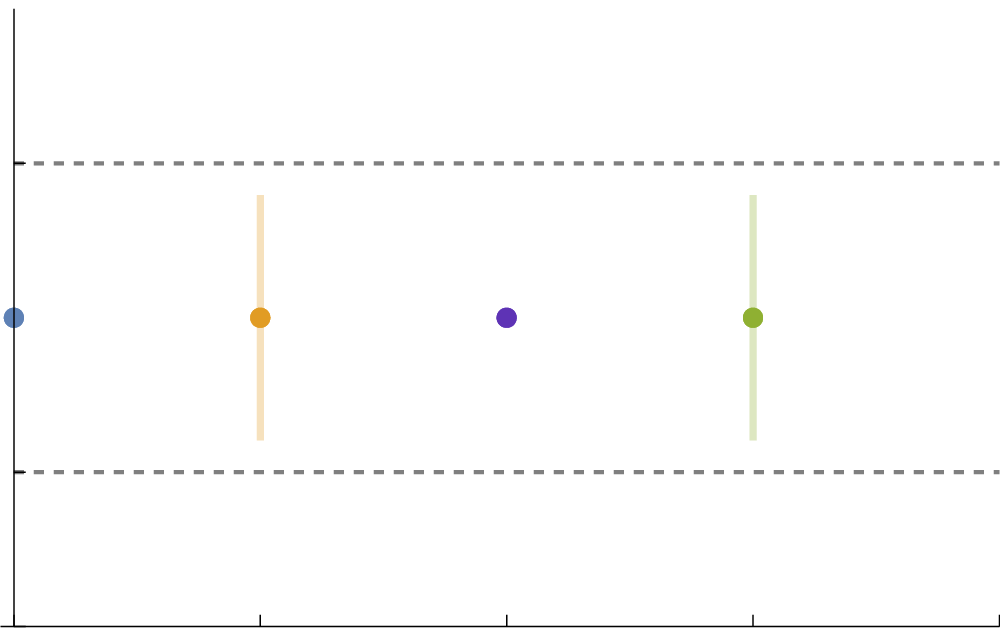}
};
\node at (-3.25,-\a+0.07) {\scriptsize{$0$}};
\node at (-1.6,-\a) {\scriptsize{$L/4$}};
\node at (0,-\a) {\scriptsize{$L/2$}};
\node at (1.6,-\a) {\scriptsize{$3L/4$}};
\node at (3.25,-\a) {\scriptsize{$L$}};

\node at (3.4,-1.65) {\small $\mathrm{Re}\,a$};
\node at (-3,2.2) {\small $\mathrm{Im}\,a$};

\node at (-3.6,1) {  \scriptsize $3\delta/2$};
\node at (-3.5,-1) {\scriptsize $\delta/2$};

\end{tikzpicture}
};

\end{tikzpicture}
\caption{Time evolution of the breather solution with initial data \eqref{eq:breatherdata} I: evolution of the poles. The left plot shows the location of the four poles $a_1$ (blue), $a_2$ (purple), $a_3$ (yellow), and $a_4$ (green) at $t=0$. In addition, the colored shadow indicates the path the poles trace as time evolves, showing that $a_1$ and $a_2$ are stationary, whereas $a_3$ and $a_4$ oscillate vertically. The right plot shows the imaginary part of the two moving poles during a full period from $t=0$ to $t=T \approx 11.83$. 
}
\label{fig:breather_poles}
\end{figure}

However, the dynamics of the spins $\bs_j$ and of the background vector $\bphi$ are not time-periodic, and correspondingly, the solution \eqref{eq:ansatzreal} of the ncIHF equation is not time-periodic. This solution is shown in Fig.~\ref{fig:breather}. We observe that at times $t=T/4+n T/2$ ($n\in \Z_{\geq 0}$), the solution has two points of non-differentiability. At these times, Corollary~\ref{cor:main} does not apply but rather guarantees a solution of the ncIHF equation on intervals with the times $\{T/4+n T/2:n\in \Z_{\geq 0}\}$ subtracted. 

\begin{figure}[t]
\centering
\begin{tikzpicture}[scale=0.9]
\def\a{3.4};
\def\b{2.5};
\def\d{2.1};
\node at (0,\d) {\includegraphics[scale=0.47]{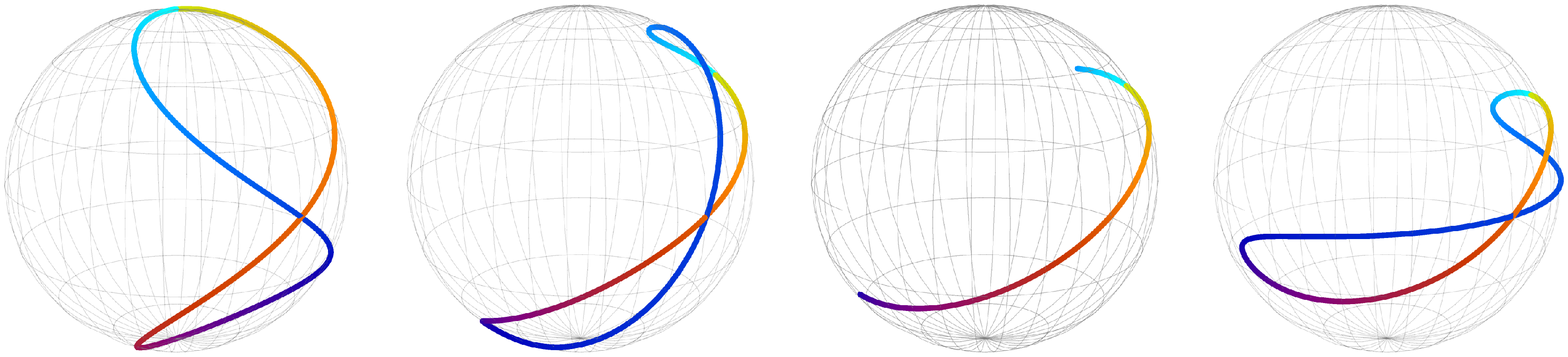}};
\node at (0,-\d) {\includegraphics[scale=0.47]{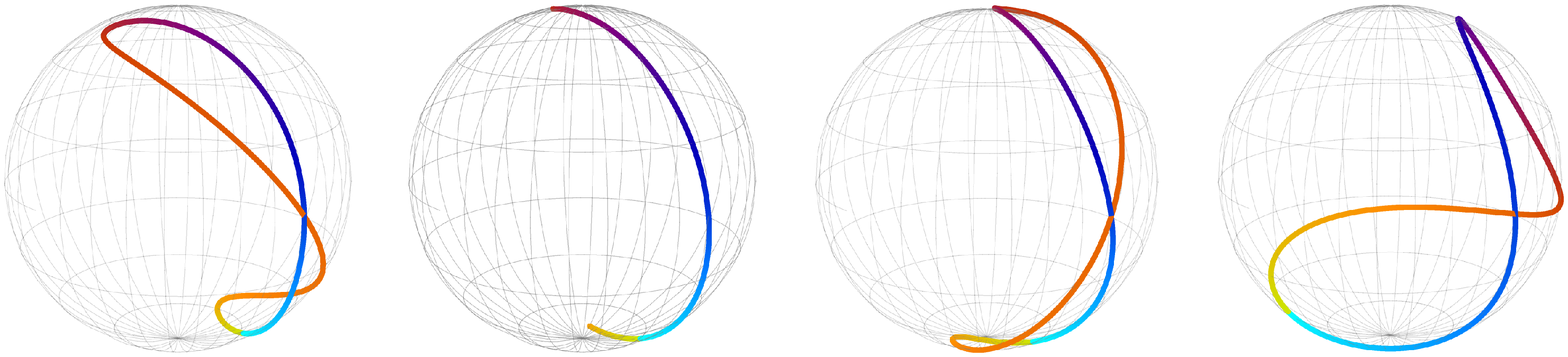}};
\node at (-0.05,-\d-\a) {\includegraphics[scale=1.3]{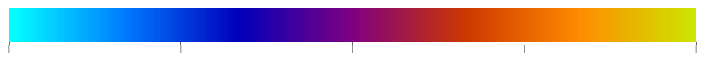}};

\node at (0,-\d-\a+0.7) {$x	$};
\node at (-5.07,-\d-\a-0.55) {\footnotesize $0$};
\node at (-2.57,-\d-\a-0.55) {\footnotesize $L/4$};
\node at (0,-\d-\a-0.55) {\footnotesize $L/2$};
\node at (2.4,-\d-\a-0.55) {\footnotesize $3L/4$};
\node at (5,-\d-\a-0.55) {\footnotesize $L$};

\draw[gray] (-8,\d+\b) -- (8,\d+\b);
\draw[gray] (-8,-\d-\b+0.2) -- (8,-\d-\b+0.2);

\foreach \x in {0,4,8,12,16}
{
\draw[gray] (-8+\x,-\d-\b+0.2) -- (-8+\x,\d+\b);
};

\node at (-6,\d+\b-0.3) {\scriptsize $t=0$};
\node at (-6+4,\d+\b-0.3) {\scriptsize $t=T/4-1/2$};
\node at (-6+8,\d+\b-0.3) {\scriptsize $t=T/4$};
\node at (-6+12,\d+\b-0.3) {\scriptsize $t=T/4+1/2$};

\node at (-6,-\d+\b-0.5) {\scriptsize $t=3T/4-1/2$};
\node at (-6+4,-\d+\b-0.5) {\scriptsize $t=3T/4$};
\node at (-6+8,-\d+\b-0.5) {\scriptsize $t=3T/4+1/2$};
\node at (-6+12,-\d+\b-0.5) {\scriptsize $t=T$};

\node at (-7.7,-\d-\b+0.6) {
\tdplotsetmaincoords{75}{90-41}
\begin{tikzpicture}[tdplot_main_coords,font=\sffamily,scale=0.4]

\draw[-latex] (0,0,0) -- (1,0,0)  node[yshift=0pt, xshift=2.3pt]  {\tiny $x$};
\draw[-latex] (0,0,0) -- (0,1,0)  node[yshift=2pt, xshift=0.5pt] {\tiny $y$};
\draw[-latex] (0,0,0) -- (0,0,1)  node[yshift=2.0pt, xshift=0pt]  {\tiny $z$};

\end{tikzpicture}
};

\end{tikzpicture}
\caption{Time evolution of the breather  solution with initial data \eqref{eq:breatherdata} II: spatial dependence of $\bu(x,t)$ at eight instances of time $t$ measured in the ``period" time $T\approx 11.83$, with colors indicating the position $x$ according to the legend on the bottom. Note that at $t=T/4$ and $t=3T/4$, when $\bu(x,t)$ is not differentiable at two points, $\bu(x,t)$ traces its image exactly twice as $x$ goes from $0$ to $L$. The plots only show one such tracing. By comparing $t=0$ and $t=T$ one sees that the image of $\bu$ is not periodic in time, in contrast to the pole and energy evolution (see Figs.~\ref{fig:breather_poles} and \ref{fig:breather_energy}). The time evolution of $\bv(x,t)$ is the reflection of $\bu$ in the $xz$-plane. The orientation of all plots is the same and indicated by the coordinate system in the bottom left corner.
}
\label{fig:breather}
\end{figure}

The energy density associated with this solution oscillates periodically in time, see Fig.~\ref{fig:breather_energy}. An explicit formula for the energy density is presented below in Section~\ref{sec:energydensity}. 

\begin{remark}
We expect that, by using the methods in \cite{krichever1995spin}, one could verify that the solution of \eqref{eq:sCM1} with the initial conditions \eqref{eq:breatherdata} and initial velocities satisfying \eqref{eq:ajdotreal} at $t=0$ exists on $[0,\infty)$ and that the poles $a_j$ are time-periodic. Then, Corollary~\ref{cor:main} would guarantee a solution of the periodic ncIHF equation on $[0,\infty)\setminus \{T/4+nT/2:n\in \Z_{\geq 0}\}$. We further expect that for a suitable notion of weak solutions of the periodic ncIHF equation, the ansatz \eqref{eq:ansatz} with the elliptic spin CM solution described above would solve the periodic ncIHF equation on $[0,\infty)$. These investigations are outside of the scope of the present paper.
\end{remark}

\subsubsection{Energy densities}\label{sec:energydensity}

\begin{figure}
\centering
\begin{tikzpicture}[scale=0.86]
\def\a{3};
\def\b{0.6};
\def\d{1.4};
\def\s{0.6};

\def\sc{0.46};
\def\scc{0.36};
\def\sccc{1.06};
\def\h{0.93};

\node at (-7,0) {
\includegraphics[scale=\sc]{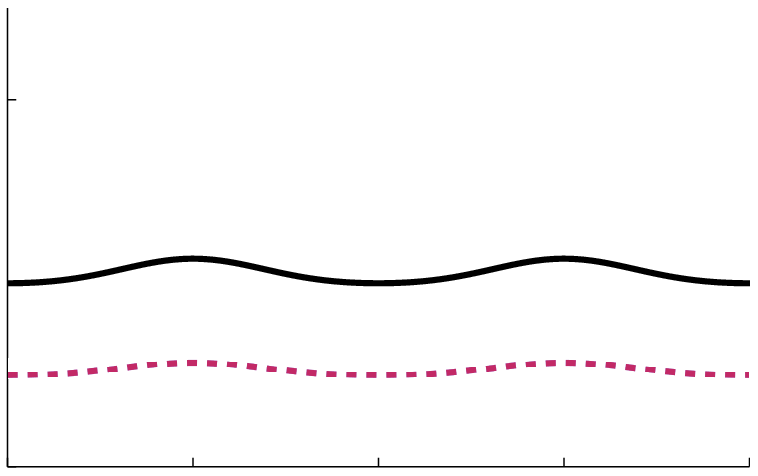}
};
\node at (-7,0) {
\begin{tikzpicture}[scale=\sccc]
\node at (1.85,-0.6) {\scalebox{0.6}{$x$}};
\node at (-1.48,1.5) {\small $\mathbf{\epsilon}$};
\node at (-1.55,-0.78) {\scalebox{\s}{$0$}};
\node at (-1.55,0.12) {\scalebox{\s}{$2$}};
\node at (-1.55,0.92) {\scalebox{\s}{$4$}};

\node at (-0.65,-0.85) {\scalebox{\s}{$L/4$}};
\node at (0.15,-0.85) {\scalebox{\s}{$L/2$}};
\node at (0.95,-0.85) {\scalebox{\s}{$3L/4$}};
\node at (1.8,-0.83) {\scalebox{\s}{$L$}};
\end{tikzpicture}
};

\node at (-2.4,0) {
\includegraphics[scale=\sc]{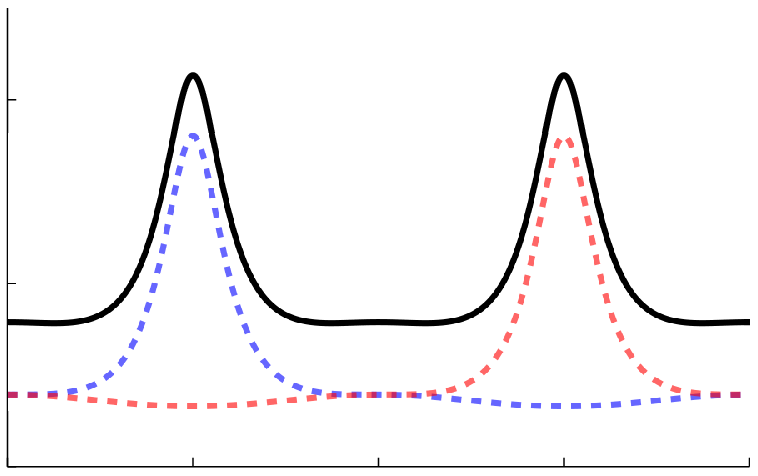}
};
\node at (-2.4,0) {
\begin{tikzpicture}[scale=\sccc]
\node at (1.85,-0.6) {\scalebox{0.6}{$x$}};
\node at (-1.48,1.5) {\small $\mathbf{\epsilon}$};
\node at (-1.55,-0.78) {\scalebox{\s}{$0$}};
\node at (-1.55,0.12) {\scalebox{\s}{$2$}};
\node at (-1.55,0.92) {\scalebox{\s}{$4$}};

\node at (-0.65,-0.85) {\scalebox{\s}{$L/4$}};
\node at (0.15,-0.85) {\scalebox{\s}{$L/2$}};
\node at (0.95,-0.85) {\scalebox{\s}{$3L/4$}};
\node at (1.8,-0.83) {\scalebox{\s}{$L$}};
\end{tikzpicture}
};

\node at (2.4,0) {
\includegraphics[scale=\sc]{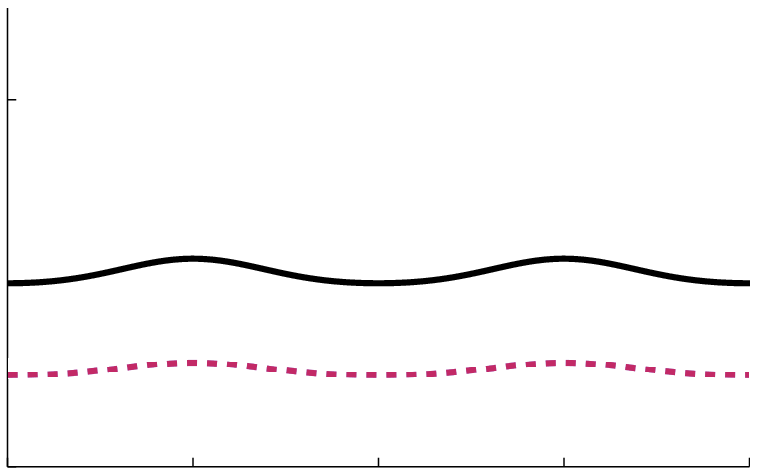}
};
\node at (2.4,0) {
\begin{tikzpicture}[scale=\sccc]
\node at (1.85,-0.6) {\scalebox{0.6}{$x$}};
\node at (-1.48,1.5) {\small $\mathbf{\epsilon}$};
\node at (-1.55,-0.78) {\scalebox{\s}{$0$}};
\node at (-1.55,0.12) {\scalebox{\s}{$2$}};
\node at (-1.55,0.92) {\scalebox{\s}{$4$}};

\node at (-0.65,-0.85) {\scalebox{\s}{$L/4$}};
\node at (0.15,-0.85) {\scalebox{\s}{$L/2$}};
\node at (0.95,-0.85) {\scalebox{\s}{$3L/4$}};
\node at (1.8,-0.83) {\scalebox{\s}{$L$}};
\end{tikzpicture}
};

\node at (7,0) {
\includegraphics[scale=\sc]{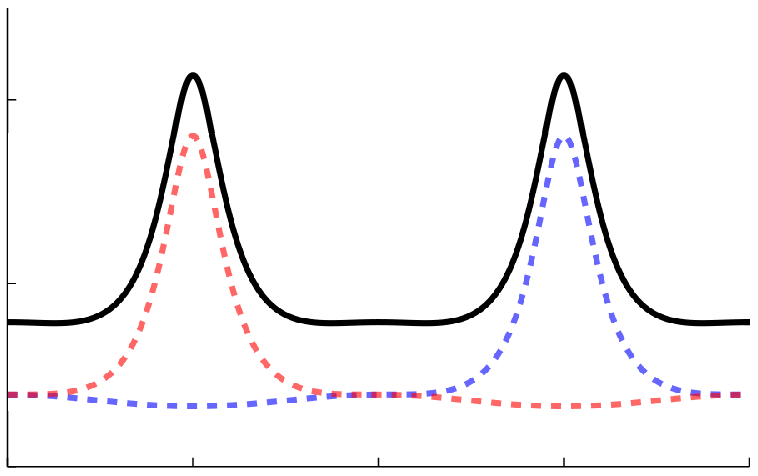}
};
\node at (7,0) {
\begin{tikzpicture}[scale=\sccc]
\node at (1.85,-0.6) {\scalebox{0.6}{$x$}};
\node at (-1.48,1.5) {\small $\mathbf{\epsilon}$};
\node at (-1.55,-0.78) {\scalebox{\s}{$0$}};
\node at (-1.55,0.12) {\scalebox{\s}{$2$}};
\node at (-1.55,0.92) {\scalebox{\s}{$4$}};

\node at (-0.65,-0.85) {\scalebox{\s}{$L/4$}};
\node at (0.15,-0.85) {\scalebox{\s}{$L/2$}};
\node at (0.95,-0.85) {\scalebox{\s}{$3L/4$}};
\node at (1.8,-0.83) {\scalebox{\s}{$L$}};
\end{tikzpicture}
};



\node at (-7,\d+\b-0.3) {\scriptsize $t=0$};
\node at (-2.4,\d+\b-0.3) {\scriptsize $t=T/4$};
\node at (2.4,\d+\b-0.3) {\scriptsize $t=T/2$};
\node at (7,\d+\b-0.3) {\scriptsize $t=3T/4$};

\end{tikzpicture}
\caption{Time evolution of the breather solution with initial data \eqref{eq:breatherdata} III: energy density at four instances of time $t$. 
At each time $t$, the total energy density $\epsilon(x,t)=\epsilon_{\bu}(x,t)+\epsilon_{\bv}(x,t)$ and the individual energy densities $\epsilon_{\bu}(x,t)$ (red) and $\epsilon_{\bv}(x,t)$ (blue) \eqref{eq:energy_uv} of the $\bu$- and the $\bv$-channels are shown. The plots illustrate that the total energy density $\epsilon(x,t)$ is periodic with period $T/2\approx 5.916$, but the $\bu$- and $\bv$-channel energy densities are periodic with period $T\approx 11.83$ only. 
At $t=T$ the energy densities are exactly the same as at $t=0$. }
\label{fig:breather_energy}
\end{figure}

It was shown in \cite[Appendix~A]{berntsonklabbers2021} that a Hamiltonian for the periodic ncIHF equation is given by
\begin{equation}\label{eq:hamiltonian}
\mathcal{H}= \int_{-\ell}^{\ell} (\epsilon_{\bu}+\epsilon_{\bv})\,\mathrm{d}x,\qquad     \left(\begin{array}{c} \epsilon_{\bu} \\ -\epsilon_{\bv} \end{array}\right)\coloneqq -\frac12 \bcu\dotcirc \cT\bcu_x    = -\frac12 \left(\begin{array}{c} \bu\cdot (T\bu_x-\tilde{T}\bv_x) \\ -\bv\cdot (T\bv_x-\tilde{T}\bu_x)        \end{array}\right),
\end{equation}
where the functions $\epsilon_\bu$ and $\epsilon_\bv$ can be interpreted as the energy densities associated with the $\bu$ and $\bv$ fields, respectively. By inserting \eqref{eq:ansatzreal} into \eqref{eq:hamiltonian} and using \eqref{eq:cTA}, \eqref{eq:ArjArk}, and \eqref{eq:constraint2}--\eqref{eq:constraint3}, a calculation similar to that in \cite[Section~5.3]{berntsonklabbers2021} gives the following result.

\begin{proposition}
The energy densities \eqref{eq:hamiltonian} associated with a real $N$-soliton solution \eqref{eq:ansatzreal} of the periodic ncIHF equation \eqref{eq:ncIHF} are given by
\begin{equation}
\label{eq:energy_uv}
\begin{split}
\epsilon_{\bu}=&\; -2\,\im\Bigg(\sum_{j=1}^N \sum_{k=1}^N\bs_j\cdot\bs_k^*\bigg(\wp_2(a_j-a_k^*+\ii\delta)\zeta_2(x-a_j+\ii\delta/2)+\frac12 f_2'(a_j-a_k^*+\ii\delta)  \bigg)\Bigg),        \\
\epsilon_{\bv}=&\; +2\,\im\Bigg(\sum_{j=1}^N\sum_{k=1}^N \bs_j\cdot\bs_k^*\bigg(\wp_2(a_j-a_k^*+\ii\delta)\zeta_2(x-a_j-\ii\delta/2)+\frac12 f_2'(a_j-a_k^*+\ii\delta)\bigg)\Bigg).
\end{split}
\end{equation}
\end{proposition}

\appendix

\section{Special functions}\label{app:elliptic}

We collect identities for the special functions needed in the main text. 

\subsection{Weierstrass elliptic functions}

We refer to \cite[Chapter~23]{DLMF} for definitions of the standard Weierstrass functions $\zeta(z)$ and $\wp(z)$. The variations of these functions we use, $\zeta_2(z)$ and $\wp_2(z)$, are defined in terms of these basic functions in \eqref{eq:zeta2} and \eqref{eq:wp2}, respectively and the function $f_2(z)$ is defined in \eqref{eq:f2}. These functions satisfy the identities 
\begin{align}
\zeta_2(z)^2=&\; \wp_2(z)+f_2(z), \label{eq:IdV}\\
\zeta_2(z-a)\zeta_2(z-b)=&\;\zeta_2(a-b)\big(\zeta_2(z-a)-\zeta_2(z-b)\big) \nonumber\\
&\;+\frac12(f_2(z-a)+f_2(z-b)+f_2(a-b)\big)+\frac{3\zeta(\ii\delta)}{2\delta}. \label{eq:Idmain} 
\end{align}
for $z,a,b\in\C$. Moreover, the following periodicity properties hold,
\begin{equation}\label{eq:realperiod}
\zeta_2(z\pm 2\ell)=\zeta_2(z)\pm\frac{\pi}{\delta},\qquad \wp_2(z\pm 2\ell)=\wp(z), \qquad f_2(z\pm 2\ell)=f_2(z)\pm \frac{2\pi}{\delta}\zeta_2(z)+\bigg(\frac{\pi}{\delta}\bigg)^2
\end{equation}
and
\begin{equation}\label{eq:imperiod}
\zeta_2(z\pm 2\ii\delta)=\zeta_2(z),\qquad \wp_2(z\pm 2\ii\delta)=\wp_2(z),\qquad f_2(z\pm 2\ii\delta)=f_2(z).
\end{equation}
Proofs of each identity \eqref{eq:IdV}--\eqref{eq:imperiod}, excepting the periodicity properties of $f_2(z)$, can be found in \cite[Appendix~A]{berntsonlangmann2021}. The periodicity properties of $f_2(z)$ follow from those of $\zeta_2(z)$ and $\wp_2(z)$ in \eqref{eq:realperiod}--\eqref{eq:imperiod} and the definition of $f_2(z)$ \eqref{eq:f2}.

The following parity properties hold as consequences of the fact that $\zeta(z)$ is an odd function, $\zeta(-z)=-\zeta(z)$, and $\wp(z)$ is an even function, $\wp(-z)=\wp(z)$, and the definitions \eqref{eq:zeta2}, \eqref{eq:wp2}, and \eqref{eq:f2}, 
\begin{equation}\label{eq:parity}
\zeta_2(-z)=-\zeta_2(z),\qquad \wp_2(-z)=\wp_2(z),\qquad f_2(-z)=f_2(z), \qquad f_2'(-z)=-f_2'(z).
\end{equation}

\subsection{Jacobi elliptic functions and elliptic integrals}

We refer to \cite[Chapter~16]{abramowitz1964} for definitions of the Jacobi functions $\sn(z|m)$ and $\cn(z|m)$. These functions are elliptic in $z$; when the \textit{elliptic parameter} $m$ satisfies $0<m<1$, the functions are real-valued for $z\in\R$. The functions satisfy the identity
\begin{equation}\label{eq:sn2cn2}
\sn^2(z|m)+\cn^2(z|m)=1.
\end{equation}
The elliptic modulus is associated with certain elliptic integrals which determine the periods of the Jacobi elliptic functions. The complete elliptic integrals of the first kind are defined by\footnote{As in Sections~\ref{subsec:parameterization}--\ref{subsec:breather}, we depart from the convention that primes indicate differentiation with respect to the argument when defining $K'=K'(m)$.}
\begin{equation}\label{eq:K}
K(m)\coloneqq \int_{0}^{\frac{\pi}{2}} \frac{\mathrm{d}\theta}{\sqrt{1-m\sin^2\theta}}
\end{equation}
and
\begin{equation}\label{eq:Kp}
K'(m)\coloneqq K(m'),\qquad m'\coloneqq 1-m.
\end{equation}

\section{Equivalent forms of the spin Calogero-Moser system} \label{app:rotation}

We prove the claim in Remark~\ref{rem:equivalent}. 

\begin{proposition}
Suppose that $\{a_j,\bs_j\}_{j=1}^N$ is a solution of the spin Calogero-Moser system \eqref{eq:sCM1} with total spin $\bS=(S^1,S^2,S^3)$ as defined in \eqref{eq:totalspin}. Let $c\in\C$ and $\mR$ be the time-dependent matrix defined by
\begin{equation}\label{eq:R}
\mR(t)\coloneqq \exp(2c\mS t)
\end{equation}
where $\mS\in \mathfrak{so}(3;\C)$ is defined to be
\begin{equation}\label{eq:S}
\mS\coloneqq \left(\begin{array}{ccc}
0 & -S^3 & S^2 \\
S^3 & 0 & -S^1 \\
-S^2 & S^1 & 0
\end{array}\right).
\end{equation}
Then, $\{\mR\bs_j,a_j\}_{j=1}^N$ is a solution of \eqref{eq:sCM1} with $\wp_2(z)\to \wp_2(z)+c$. 
\end{proposition}

\begin{proof}
By construction, $\mR\in \mathrm{SO}(3;\C)$ and is thus invertible with inverse $\mR^{-1}=\mR^{\trans}=\exp(-c\mS t)$. Under the transformations $\bs_j\to \mR\bs_j$ ($j=1,\ldots,N$) and $\wp_2(z)\to \wp_2(z)+c$,  \eqref{eq:sCM1} becomes
\begin{subequations}\label{eq:sCM3}
\begin{align}
\ddot{a}_j=&\; -2\sum_{k\neq j}^N (\mR\bs_j)\cdot(\mR\bs_k) \wp_2'(a_j-a_k), \label{eq:sCM3a} \\
\mR\dot{\bs}_j+\dot{\mR}\bs_j=&\;  -2\sum_{k\neq j}^N\ (\mR\bs_j)\wedge(\mR\bs_k)(\wp_2(a_j-a_k)+c). \label{eq:sCM3s}
\end{align}
\end{subequations}
We show that \eqref{eq:sCM3} holds if and only if \eqref{eq:sCM1} holds. Using the invariance of the dot product under orthogonal transformations, $(\mR\bs_j)\cdot(\mR\bs_k)=\bs_j\cdot\bs_k$, we see that \eqref{eq:sCMa} and \eqref{eq:sCM3a} are equivalent.
Cross products transform under orthogonal transformations as $(\mR\bs_j)\wedge(\mR\bs_k)=\mR(\bs_j\wedge\bs_k)$; using this fact and multiplying by $\mR^{-1}$ in \eqref{eq:sCM3s} gives
\begin{equation}
\dot{\bs}_j+\mR^{-1}\dot{\mR}\bs_j=  -2\sum_{k\neq j}^N\ \bs_j\wedge\bs_k(\wp_2(a_j-a_k)+c)
\end{equation}
We observe that $\mR$ satisfies the differential equation
\begin{equation}
\mR^{-1}\dot{\mR}=2c\mS,
\end{equation}
and that
\begin{equation}\label{eq:Ssj}
\mS\bs_j=\bS\wedge\bs_j=\sum_{k=1}^N \bs_k\wedge\bs_j=-\sum_{k\neq j}^N \bs_j\wedge\bs_k.
\end{equation}
It follows that
\begin{equation}
\dot{\bs}_j-2c\sum_{k\neq j}^N \bs_j\wedge\bs_k    =  -2\sum_{k\neq j}^N\ \bs_j\wedge\bs_k(\wp_2(a_j-a_k)+c),
\end{equation}
which becomes \eqref{eq:sCMs} after cancellations. 
\end{proof}

\begin{remark}
The matrix $\mR$ in \eqref{eq:R} can be written explicitly as 
\begin{equation}\label{eq:R2}
\mR(t) =  I+\frac{\sin\big(2c\sqrt{\bS\cdot\bS}t\big)}{\sqrt{\bS\cdot\bS}}\mS+\frac{1-\cos\big(2c\sqrt{\bS\cdot\bS}t\big)}{\bS\cdot\bS}\mS^2,
\end{equation}
where $I$ is the $3\times 3$ identity matrix and the choice of branch in $\sqrt{\bS\cdot\bS}$ is immaterial. This is equivalent to Rodrigues' formula for the exponential map from $\mathfrak{so}(3;\C)$ to $\mathrm{SO}(3;\C)$ \cite{murray1994}.
\end{remark}

\section*{Acknowledgements}

BKB thanks Jonatan Lenells for helpful discussions and collaboration on closely related projects. BKB and RK gratefully acknowledge valuable input from Edwin Langmann at several stages of this project. The work of BKB was supported by the Olle Engkvist  Foundation, Grant 211-0122.

\bibliographystyle{unsrt}

\bibliography{BKsubmit}

\end{document}